\documentclass[11pt,a4paper]{article}

\usepackage[utf8x]{inputenc}
\usepackage[T1]{fontenc}

\usepackage{amsmath,amssymb,amsfonts,amsthm,dsfont}
\usepackage{comment}
\usepackage{fullpage}

\usepackage[numbers,sort&compress]{natbib}

\usepackage{pgfplots}
\pgfplotsset{compat=1.3}

\usepackage{tikz}
\usetikzlibrary{positioning,backgrounds,patterns,calc,fit,shapes,matrix,arrows}
\tikzstyle{vertex}=[circle,draw=black,minimum size=8pt,inner sep=1pt]
\tikzstyle{vertex2}=[circle,draw=black,minimum size=15pt,inner sep=2pt]
\tikzstyle{edge}=[]
\tikzstyle{ypath}=[ultra thick]
\tikzstyle{dottedEdge}=[dotted,thick]
\tikzstyle{small-vertex}=[circle,draw=black,minimum size=6pt,inner sep=0pt,fill=white]
\tikzstyle{thinedges}=[draw=gray!30]
 \tikzstyle{boxes}=[draw,thick, rounded corners=3mm,text width=2.7cm,align=center,text opacity=1,fill opacity=1,fill=white]
\tikzstyle{unk}=[fill=gray!25!white]
\tikzstyle{nodest}=[circle,draw,minimum size=0.6cm,inner sep=1pt]
\tikzset{linemark/.style =   {line width= 4pt, color=#1}}

\tikzstyle{normalline} = [line width=.8pt]

\newcommand{\Agent}{\text{Agent}}
\newcommand{\agent}{\text{agent}}

\usepackage{mathtools}
\usepackage{xspace}

\usepackage{xhfill}

\newtheorem{theorem}{Theorem}[section]

\newtheorem{proposition}[theorem]{Proposition}
\newtheorem{myclaim}{Claim}

\newtheorem{lemma}[theorem]{Lemma}

\theoremstyle{definition}
\newtheorem{example}[theorem]{Example}
\newtheorem{definition}[theorem]{Definition}

\newcommand{\indif}{\ensuremath{\sim}}

\usepackage{graphicx}
\usepackage{paralist}

\usepackage[margin=0pt,position=t]{subcaption}

\usepackage{booktabs}
\usepackage{multirow}

\usepackage{xcolor}
\definecolor{dargray}{rgb}{0.18, 0.18, 0.18}
\definecolor{darkgreen}{rgb}{0.01,0.6,0.1}
\definecolor{lightrose}{rgb}{0.996,0.75,0.793}
\definecolor{rose}{cmyk}{0.75, 0.75, 0,0}
\definecolor{winered}{rgb}{0.6,0.1,0.1}
\definecolor{lightyellow}{rgb}{1, 1, 0.6}
\definecolor{darkyellow}{rgb}{.99, .87, 0.04}
\definecolor{transparent}{rgb}{1,1,1}
\definecolor{lightlightgray}{rgb}{0.88, 0.88, 0.88}
\definecolor{lightgray}{rgb}{0.8, 0.8, 0.8}
\definecolor{lightblue}{rgb}{0.527,0.805,0.977}
\definecolor{lightgreen}{rgb}{.74,1,0}

\newcommand{\wred}[1]{{\color{winered}#1}}

\usepackage[
bookmarks=false,
breaklinks=true,
colorlinks=true,
linkcolor=blue,
citecolor=winered,
urlcolor=purple,
pdfpagelayout=SinglePage,
pdfstartview=Fit,
pagebackref=true
]{hyperref}
\usepackage[all]{hypcap}

\usepackage[noend,ruled,linesnumbered]{algorithm2e}

\SetCommentSty{mycommfont}
\SetKwComment{Comment}{$\triangleright$\ }{}

\SetAlFnt{\small}
\SetAlCapFnt{\small}
\SetAlCapNameFnt{\small}
\SetAlCapHSkip{0pt}
\IncMargin{-\parindent}
\SetKwFunction{computeStableMatching}{ComputeStableMatching}
\SetKwFunction{insertMostPreferred}{InsertMostPreferred}
\SetKwFunction{smartDelete}{SmartDelete}
\newcommand{\AND}{{\normalfont \textbf{and}}\xspace}
\newcommand{\OR}{{\normalfont \textbf{or}}\xspace}

\newcommand{\CcomputeSM}{\texttt{ComputeStableMatching}}
\newcommand{\CinsertMP}{\texttt{InsertMostPreferred}}
\newcommand{\CsmartD}{\texttt{SmartDelete}}

\usepackage[capitalise,nameinlink]{cleveref}
\crefname{subsection}{Section}{Sections}
\crefname{section}{Section}{Sections}

\crefname{table}{Table}{Tables}
\crefname{figure}{Figure}{Figures}
\crefname{algorithm}{Algorithm}{Algorithms}
\crefname{theorem}{Theorem}{Theorems}
\crefname{conjecture}{Conjecture}{Conjectures}
\crefname{definition}{Definition}{Definitions}
\crefname{corollary}{Corollary}{Corollary}
\crefname{proposition}{Proposition}{Propositions}
\crefname{observation}{Observation}{Observations}
\crefname{lemma}{Lemma}{Lemmas}
\crefname{example}{Example}{Examples}
\crefname{reduction}{Reduction}{Reductions}
\crefname{appendix}{Appendix}{Appendices}
\crefname{claim}{Claim}{Claims}
\crefname{line}{line}{lines}
\crefname{myclaim}{Claim}{Claims}

\newcommand{\Pot}{\mathcal{P}}
\newcommand{\match}[3]{\ensuremath{\{#2, #3\}}}

\newcommand{\SR}{\textsc{Stable Roommates}\xspace}
\newcommand{\tiesc}{tie-sen\-si\-tive single-crossing\xspace}
\newcommand{\tiescness}{tie-sensitive single-crossingness\xspace}
\newcommand{\Tiescness}{Tie-sensitive single-crossingness\xspace}

\newcommand{\border}{\ensuremath{\blacktriangleright}}

\usepackage{etoolbox}

\newcommand{\decprob}[3]{
  \begin{center}%
    \begin{minipage}{0.9\linewidth}%
      \textsc{#1}\\[0.2ex]
      \textbf{Input:} #2\\[0.2ex]
      \textbf{Question:} #3
    \end{minipage}%
  \end{center} }

\newcommand{\pointer}{\ensuremath{p}}

\newcommand{\clR}[1]{\textcolor{black}{#1}}
\newcommand{\clB}[1]{\textcolor{black}{#1}}
\newcommand{\clG}[1]{\textcolor{black}{#1}}
\newcommand{\clO}[1]{\textcolor{black}{#1}}

\newenvironment{profile}[1]{
\begin{quote}
$\begin{array}{#1}
}{ 
\end{array}$
\end{quote}
}
\usepackage{authblk}

\newcommand{\thxtxtA}{A preliminary version of this work appeared in the Proceedings of the 5th International
Conference on Algorithmic Decision Theory~(ADT\ '17)~\cite{BCFN17}, volume 10576 of LNCS, pages 315--330,
Springer, 2017. This full version contains (more) proof details for \cref{prop:relation-sc-tsc}, \cref{prop:tiesless-incomp-sp-sc-np-c}, \cref{thm:complete-ties-nsp-qudratic}, \cref{incomplete-ties-sp-sc:np-c}, and \cref{complete-sc:p}. Furthermore, the reduction used for our main result (\cref{incomplete-ties-sp-sc:np-c}) was replaced by a completely new reduction showing NP-hardness for the case of narcissistic, single-peaked,
and single-crossing preferences (the preferences in the previous reduction were not narcissistic).}
\newcommand{\thxtxtB}{Most of the work was done while all authors were with TU~Berlin,
  with some additional work done while Jiehua Chen was with Ben-Gurion University, Israel, and with University of Warsaw, Poland.}
\title{Stable Roommates with Narcissistic, Single-Peaked, and Single-Crossing Preferences\thanks{\thxtxtA} \thanks{\thxtxtB}}

\author[1]{Robert Bredereck}
\author[2]{Jiehua Chen}
\author[1]{Ugo Paavo Finnendahl}
\author[1]{Rolf Niedermeier}

\affil[1]{
  TU Berlin, Berlin, Germany\protect\\
  \texttt{\{robert.bredereck, finnendahl, rolf.niedermeier\}@tu-berlin.de}
  }
\affil[2]{
  TU Wien, Vienna, Austria\protect\\
  \texttt{jiehua.chen@tuwien.ac.at}}
\date{}

\begin{document}

\maketitle 

\begin{abstract}
  The classical \SR problem is to decide whether there exists
  a matching of an even number of agents such that no two agents which are not matched to each other would prefer to be with each other rather than with their respectively 
assigned partners. 
  We investigate \SR with complete (i.e., every agent can be matched with any other agent) or incomplete preferences,
  with ties (i.e., two agents are considered of equal value to some agent) or without ties. 
  It is known that in general allowing ties makes the problem NP-complete.
  We provide algorithms for \SR that are, compared to those in the literature, more efficient when the input preferences are complete and have some structural property, such as being narcissistic, single-peaked, and single-crossing.
  However, when the preferences are incomplete and have ties,
  we show that being single-peaked and single-crossing does not reduce the computational complexity---\SR remains NP-complete.
\end{abstract}

\section{Introduction}\label{sec:intro}

Given $2n$ agents, each having preferences with regard to how suitable the other agents are as potential partners, 
the \SR problem is to decide whether there exists a matching, i.e., a set of disjoint pairs of the agents, 
without inducing a \emph{blocking pair}.
A blocking pair consists of two agents that are not matched to each other but prefer to be with each other rather than with their assigned partners. 
A matching without blocking pairs is called a \emph{stable matching}.

\SR was introduced by \citet{GaleShapley1962} in the 1960's and has been studied extensively since then~\cite{Knuth1997,Irving1985,Ronn1990,RothSotomayor1990,IrvingManlove2002,IwamaMiyazaki2008}.
While it is quite straightforward to see that stable matchings may not always exist, 
it is not trivial to see whether an existing stable matching can be found in polynomial time,
even when the input preference orders are \emph{complete} and \emph{do not contain ties} (i.e., each agent can be a potential partner to each other agent, and no two agents are considered to be equally suitable as a partner).
For the case without ties, 
\citet{Irving1985} and \citet{GusfieldIrving1989} provided $O(n^2)$-time algorithms 
to decide the existence of a stable matching and to find one if it exists---for complete preferences and for incomplete preferences, respectively. 
Deciding whether a given instance admits a stable matching is NP-complete~\cite{Ronn1990} when the given preferences are complete but may contain ties.

Solving \SR has many applications, such as matching students with each other to accomplish a homework project or users in a P2P file sharing network, 
assigning co-workers to two-person offices,
partitioning players in two-player games,
or finding receiver-donor pairs for organ transplants~\cite{KuLiMa1999,LeMaViGaReMo2006,RoSoUn2005,RoSoUn2007,MaOM2014}.
In such situations, the students, the people, or the players, who we jointly refer to as \emph{agents},
typically have certain \emph{structurally restricted} preferences concerning which other agents might be their best partners.
For instance,
when assigning roommates,
each agent may have an ideal room temperature and may prefer to be with another agent with the same penchant.
Such preferences are called \emph{narcissistic}. 
Moreover, if we order the agents according to their ideal room temperatures,
then it is natural to assume that each agent~$z$ prefers to be with an agent~$x$ rather than with another agent~$y$ if $z$'s ideal temperature is closer to~$x$'s than to~$y$'s.
These kind of preferences are called \emph{single-peaked}~\cite{Hotelling1929,Black1958,Coombs1964}.
Single-peakedness is used to model agents' preferences where 
there is a criterion, e.g., room temperature, 
that can be used to obtain a linear order of the agents
such that each agent's preferences over all agents along this order 
are strictly increasing until they reach the peak---their ideal partner---and then strictly decreasing.
Single-peakedness is a popular concept 
with prominent applications in voting contexts. 
If the input preferences are complete and contain no ties, then testing whether they are single-peaked can be done in linear time~\cite{BartholdiTrick1986,BaHa2011,DoiFal1994,EscLanOez2008}. 

Another possible restriction on the preferences is the \emph{single-crossing} property, 
which was originally proposed to model individuals' preferences on income taxation~\cite{Mirrlees1971,Roberts1977}.
To illustrate this, assume that the agents are increasingly ordered from 
left to right according to their income level.
Then, one may observe that the agents' roommates preferences towards two arbitrary agents, say $x$ and $y$, are such that the agents on the left prefer $x$ to $y$ while the agents on the right prefer $y$ to $x$.
This kind of preferences is not necessarily single-peaked as people do not always prefer to have a roommate of similar income level. 
If the input preferences are complete and contain no ties, then testing whether they are single-crossing can be done in polynomial time~\cite{DoiFal1994,ElkFalSli2012,BreCheWoe2013a}. 
We refer to \citet{BreCheWoe2016} and \citet{ElLaPe2017} for numerous references on single-peakedness and single-crossingness.

\paragraph{Related work.}

\citet{BartholdiTrick1986} studied \SR with narcissistic and single-peaked preferences.
They showed that for the case with linear orders (i.e., complete and without ties) a \SR instance with $2n$~agents \emph{always} admits a \emph{unique} stable matching, and they claimed an $O(n)$-time algorithm to find this matching~\cite[Section 3]{BartholdiTrick1986}.
However, we will show in \cref{ex:BT-n^2} that there are instances for which their algorithm runs in $O(n^2)$~time.
In terms of preference structures in the stable matching setting, using a connection to narcissistic
single-crossing preference profiles~(see \cref{sub:profile+properties} for the corresponding definition) 
and \emph{semi-standard Young tableaux}~\cite{Stanley1999},
\citet{CheFin2018} counted the number of narcissistic preference profiles that are also either
single-peaked or single-crossing.

\paragraph{Paper structure and our contributions.}
We study the computational complexity of \SR for structured preferences when incompleteness and ties are allowed.
In particular, we explore how the specific preference structures help in guaranteeing the existence of stable matchings and in designing 
efficient algorithms for finding a stable matching, even when the input preferences may be incomplete or contain ties.

\newcommand{\citeronn}{$^\triangle$\xspace}
\newcommand{\citeirv}{$^\diamondsuit$\xspace}
\newcommand{\citegi}{$^\spadesuit$\xspace}
\newcommand{\citebt}{$^\heartsuit$\xspace}
\newcommand{\citebcfn}{$^\star$\xspace}
\newcommand{\other}[1]{{\color{gray}#1}}
\newcommand{\mynew}[1]{{\color{winered}\boldmath{#1}}}
\begin{table}[t]
\centering
\caption{Complexity of \SR for restricted domains: narcissistic, single-peaked, and single-crossing preferences.
  ``nar'' means ``narcissistic'',
  ``sp'' means ``single-peaked'',
  ``sc'' means ``single-crossing'',
  ``always'' means that there is always a stable matching,
  and ``unique'' means that there there is a unique stable matching.
  Entries marked with \citeirv are from \citet{Irving1985}.
Entries marked with \citegi are from \citet{GusfieldIrving1989}.
Entries marked with \citeronn are from \citet{Ronn1990}.
Entries marked with \citebt are from \citet{BartholdiTrick1986}.
Entries marked boldfaced with a reference to the corresponding theorem~[T.~$x$] or proposition~[P.~$x$]
are new results shown in this paper.
Note that our hardness results for single-crossing preferences hold for the more restricted ``\tiesc'' variant.}
\label[table]{table:results}
\resizebox{\textwidth}{!}{
\begin{tabular}{lp{10ex}p{10ex}llll}
 \toprule
  & \multicolumn{2}{c}{without ties} && \multicolumn{2}{c}{with ties}\\
  & \multicolumn{1}{c}{complete}  & \multicolumn{1}{c}{incomplete} &&  \multicolumn{1}{c}{complete} &  \multicolumn{1}{c}{incomplete}\\
  \midrule
  no restriction & \multicolumn{1}{c}{$O(n^2)$\citeirv} & \multicolumn{1}{c}{$O(n^2)$\citegi} &&  \multicolumn{1}{c}{NP-c\citeronn} & \multicolumn{1}{c}{NP-c\citeronn}\\
  nar. \& sp & \multicolumn{1}{l}{$O(n^2)$, always, unique\citebt} & \multicolumn{1}{c}{$O(n^2)$\citegi} && \multicolumn{1}{c}{\mynew{$O(n^2)$, always\ [T.~\ref{thm:complete-ties-nsp-qudratic}]}} & \multicolumn{1}{c}{\mynew{\textbf{NP-c}\ [T.~\ref{incomplete-ties-sp-sc:np-c}]}}\\
  nar. \& sc & \multicolumn{1}{l}{\mynew{$O(n^2)$, always, unique\ [P.~\ref{complete-sc:p}]}} & \multicolumn{1}{c}{$O(n^2)$\citegi} && \multicolumn{1}{c}{\mynew{$O(n^2)$, always\ [P.~\ref{complete-sc:p}]}} & \multicolumn{1}{c}{\mynew{\textbf{NP-c}\ [T.~\ref{incomplete-ties-sp-sc:np-c}]}}\\
  nar. \& sp \& sc & \multicolumn{1}{l}{$O(n^2)$, always, unique\citebt\ \mynew{[P.~\ref{complete-sc:p}]}} & \multicolumn{1}{c}{$O(n^2)$\citegi} && \multicolumn{1}{c}{\mynew{$O(n^2)$, always\ [P.~\ref{complete-sc:p}]}} & \multicolumn{1}{c}{\mynew{\textbf{NP-c}\ [T.~\ref{incomplete-ties-sp-sc:np-c}]}}\\
\bottomrule
\end{tabular}
}
\end{table}

In \cref{sec:prelim},
we discuss natural generalizations of the well-known single-peaked and single-crossing preferences (that were originally introduced for linear orders) for incomplete preferences with ties.
In \cref{sec:complete}, we show that for complete preference orders,
structurally restricted preferences such as being narcissistic and single-crossing or
being narcissistic and single-peaked
guarantee the existence of stable matchings.  
Moreover, we demonstrate that the known algorithm of \citet{BartholdiTrick1986} can be extended to
always find a stable matching in two new cases:
The algorithm works
\begin{compactenum}[(1)]
  \item when the preferences are complete, narcissistic, single-crossing, and may contain ties
  as well as
  \item when the preferences are complete, narcissistic, single-peaked, and may contain ties.
\end{compactenum}
In \cref{sec:incomplete} we study the case where the preferences are incomplete
and may contain ties,
and prove that \SR becomes NP-complete, even when the preferences are narcissistic, single-peaked, and (tie-sensitive) single-crossing. 
Our results, together with those from related work, are surveyed in \cref{table:results}.
We conclude in \cref{sec:conclude} with some open questions.


\section{Fundamental concepts and basic observations}\label{sec:prelim}

In this section, we introduce fundamental concepts and notions, arising from stable matchings and structured preferences, and we make some crucial observations regarding relations between the different structured preferences.

\subsection{Preferences, acceptable sets, and acceptability graphs}
Let $V=\{1,2,\ldots, 2n\}$ be a set of $2n$~agents.
A \emph{(preference) profile}~$\Pot=(\succ_i)_{i\in V}$ is a collection of the preference orders of the agents from $V$.
Here, the \emph{preference order~$\succeq_i$} of each agent~$i\in V$ is a
weak order on a subset~$V(\succeq_i)\subseteq V$ of agents
that $i$ finds \emph{acceptable} as a partner.
The set~$V(\succeq_i)$ is called the \emph{acceptable set} of~$i$.
Recall that a \emph{weak order}~$\succeq$ on a set~$X$ is a transitive (i.e., $x\succeq y$
and $y\succeq z$ imply $x\succeq z$)
and complete~(i.e., $x\succeq y$ or $y\succeq x$ for all $x,y\in X$) binary relation on~$X$~(see \cite[Chapter 4]{BraFis2002} or \cite[Chapter 1.2.2]{BraConEndLanPro2016}).
For instance, the following binary relation~$\succeq$ with 
\begin{align}\label{eq:weak-order-ex}
  \succeq = \{(1,1),(2,2),(3,3),(1,2), (2,1), (1,3), (2,3)\}
\end{align}
is a weak order on~$\{1,2,3\}$. 
Note that the completeness of weak orders implies reflexivity.
For each agent~$i$ and each two acceptable agents~$x,y\in V(\succeq_i)$ of~$i$, the expression~``$x \succeq_i y$'' 
means that agent~$i$ weakly prefers~$x$ over~$y$
(i.e.,~$i$ finds that $x$ is at least as good as~$y$).
We use \emph{$\succ_i$} to denote the asymmetric part of~$\succeq_i$ 
(i.e., $x\succeq_i y$ and $\neg (y\succeq_i x)$, meaning that $i$ strictly prefers $x$ to $y$)
and \emph{$\indif_i$} to denote the symmetric part of $\succ_i$
(i.e., $x\succeq_i y$ and $y \succeq_i x$, meaning that $i$ values $x$ and $y$ equally).

Since a preference order~$\succeq$ can be decomposed into an asymmetric part~$\succ$ and a symmetric part~$\indif$, in the following when we illustrate a preference order, we only describe the~$\succ$ part and the ``relevant'' $\indif$~part.
For instance, the preference order~$\succeq$ as described in~\eqref{eq:weak-order-ex} will be depicted as follows:
\begin{align*}
  \succeq = 1 \indif 2 \succ 3.
\end{align*}
Observe that we omitted the relations~$1\sim 1$, $2\sim 2$, and $3\sim 3$ for the sake of readability.

We assume that the acceptability relation between each two agents is symmetric, i.e., for each two distinct agents~$i$ and~$j$ it holds that ``$i$ finds $j$ acceptable if and only if $j$ finds $i$ acceptable'',
as otherwise $i$ and $j$ will never be partners of each other.
Formally, this means that $i\in V(\succeq_j)$ if and only if $j \in V(\succeq_i)$.
Moreover, since an agent that is not acceptable to any other agent will never obtain a partner,
we also assume that for each agent~$i$ there is at least another agent~$j\neq i$ with $i\in V(\succeq_j)$.

We note that although in \textsc{Stable Roommates}
an agent cannot be matched to itself, 
it may still make sense to include an agent~$x$ in its own acceptable set, i.e., $x \in V(\succeq_x)$,
for instance when the preferences of~$x$ are based on how close or similar
agents are to the ideal partner of~$x$ and a partner which is ``identical'' to~$x$ itself is an ideal partner of~$x$. 
This has no implication at this point, because no agent can be matched to itself and ``identical clones'' do not exist.
Allowing this, however, drastically simplifies definitions of restricted preferences in \cref{sub:profile+properties}, in particular when narcissistic preferences are involved.
We call an agent~$x$ a \emph{most acceptable} agent of another agent~$y$ if for all $z \in V(\succeq_y)\setminus \{x,y\}$ it holds that $x \succeq_y z$,
where $\succeq_y$ denotes the preference order of~$y$.
Note that an agent can have more than one most acceptable agent.

Let $X \subseteq V$ and $Y \subseteq V$ be two disjoint sets of agents and $\succeq$ be a binary relation over~$V\times V$.
To simplify notation, 
by $X\succeq Y$, we mean that for each two agents~$x$ and $y$ with $x\in X$ and $y \in Y$ 
it holds that $x\succeq y$.
Analogously, by $X \succ Y$ and $X \indif Y$ we mean that
for each two agents~$x$ and $y$ with $x\in X$ and $y \in Y$ 
it holds that $x\succ y$ and $x \indif y$, respectively.
For each binary relation symbol~$\star \in \{\succeq, \succ, \indif\}$,
we use $X \star y$ and $y\star X$ as shortcut for $X \star \{y\}$ and $\{y\} \star X$,
respectively.

To visualize which agent is considered as acceptable by an agent we 
introduce the notion of \emph{acceptability graph}s.
An \emph{acceptability graph}~$G$ for a set~$V$ of agents is
an undirected graph without loops, where an edge signifies that two distinct agents find each other acceptable.
We use~$V$ to also denote the vertex set of~$G$. 
Formally, for each agent~$i\in V$, there is a vertex~$i$ corresponding to agent~$i$.
There is an edge~$\{i,j\}$ in $G$ if $i \in V(\succeq_j) \setminus \{j\}$ and $j \in V(\succeq_i)\setminus \{i\}$.
As already discussed,  we assume without loss of generality that $G$ does not contain isolated vertices as otherwise the corresponding agents will never be able to obtain a partner.
We illustrate two prominent special cases of acceptability graphs in \cref{fig:acceptability-graph}.

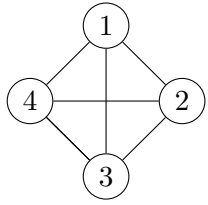
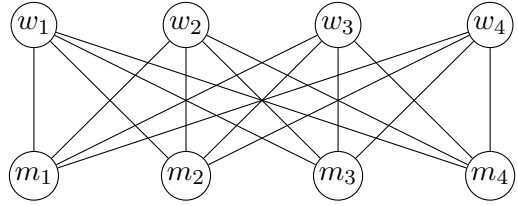
\begin{figure}[t]
  \begin{subfigure}[t]{0.39\textwidth}
    \centering
    \begin{tikzpicture}
    \foreach \i in {1,2,3,4} {
      \node[nodest] (n\i) at (-\i*90+180:1) {$\i$};
    }

      \foreach \from/\to in {n1/n2,n1/n3,n1/n4,n2/n3,n2/n4,n3/n4,n3/n4}
        \draw (\from) -- (\to);
    \end{tikzpicture}
    \caption{The underlying acceptability graph of a \SR instance with complete preferences,
                    where \emph{any} two distinct agents may find each other acceptable.}
    \label{fig:acceptability-graph:unrestricted}
  \end{subfigure}
  \hspace{0.05\textwidth}
  \begin{subfigure}[t]{0.53\textwidth}
    \centering
    \begin{tikzpicture}
      \node[nodest] (w1) at (1,2) {$w_1$};
      \node[nodest] (w2) at (3,2) {$w_2$};
      \node[nodest] (w3) at (5,2) {$w_3$};
      \node[nodest] (w4) at (7,2) {$w_4$};
      \node[nodest] (m1) at (1,0) {$m_1$};
      \node[nodest] (m2) at (3,0) {$m_2$};
      \node[nodest] (m3) at (5,0) {$m_3$};
      \node[nodest] (m4) at (7,0) {$m_4$};

      \foreach \from in {w1,w2,w3,w4}
        \foreach \to in {m1,m2,m3,m4}
          \draw (\from) -- (\to);
    \end{tikzpicture}
  \caption{The underlying acceptability graph of a classic \textsc{Stable Marriage} instance,
    which is always bipartite.
    In such an instance, each woman from the top row can only be
           matched with a man from the bottom row,
and the converse.}\label{fig:acceptability-graph:stable-marriage}
  \end{subfigure}
  \caption{Acceptability graphs of two special cases of \SR.}\label{fig:acceptability-graph}
\end{figure}

\subsection{Blocking pairs and stable matchings}
Given a preference profile~$\Pot$ for a set~$V$ of agents,
a \emph{matching}~$M\subseteq E(G)$ 
is a subset of disjoint pairs~$\{x,y\}$ of agents with $x\neq y$ (or edges in $E(G)$), 
where $E(G)$ is the set of edges in the corresponding acceptability graph~$G$.
Slightly abusing notation, given an agent~$x\in V$, if there exists an agent~$y\in V$
with $\{x,y\}\in M$, then we let \emph{$M(x)$} $\coloneqq y$ and say that $x$ and $y$ are partners of each other (under $M$); otherwise we let $M(x)\coloneqq \bot$. 
We say that a pair $\{x,y\}$ is \emph{unmatched} (under $M$) if $\{x,y\}\notin M$.
A matching~$M$ is \emph{perfect} if every agent is assigned a partner by~$M$.
An unmatched pair~$\{x,y\}\in E(G)\setminus M$ \emph{is blocking}~$M$
if the pair ``prefers'' to be matched with each other rather than staying in their current state, i.e.,
it holds that 
\begin{align*}
( M(x)=\bot \vee y\succ_x M(x) ) \wedge ( M(y)=\bot \vee x \succ_y M(y))\text{.}
\end{align*}

A matching~$M$ is \emph{stable} if no unmatched pair is blocking $M$.
When the preferences may contain ties, our stability concept is sometimes referred to as \emph{weak stability} in the literature to distinguish from two other stability concepts,
called \emph{strong stability} and \emph{super stability}~\cite{GusfieldIrving1989}.
When the preferences do not contain ties, all these three stability concepts are equivalent.
In this work, we only focus on \emph{weak stability}.
For the sake of brevity, we thus simply use \emph{stability} to refer to \emph{weak stability}. 

\begin{example}
\label{ex:running}
  Consider the following profile: 
  \begin{profile}{rccccccc}
   \text{agent } 1 \colon & 1 & \succ & 2  & \succ & 3 & \succ & 4,\\
   \text{agent } 2 \colon & 2 & \succ & 3 & \succ & 1 & \succ & 4,\\
   \text{agent } 3 \colon & 3 & \succ & 2 & \indif & 4 & \succ & 1,\\
   \text{agent } 4 \colon & 4 & \succ & 3 & \succ & 2 & \succ & 1.\\
 \end{profile}
  It admits exactly two stable matchings:~$M_1=\{\{1,2\}, \{3,4\}\}$,
  and $M_2=\{\{1,4\}$, $\{2,3\}\}$; both are perfect. 
  However, if agent~$3$ changes its preference order to $3\succ 1 \succ 2 \succ 4$,
  then the resulting profile does not admit any stable matching:
  One can check that for each matching, any agent~$i$, $1\le i \le 3$, that is matched to agent~$4$ will form a blocking pair
  together with the agent that is at the third position of the preference order of~$i$.
\end{example}

\noindent We investigate the computational complexity of the following stable matching problem.
\decprob{Stable Roommates}
{A preference profile~$\Pot$ for a set~$V=\{1,2,\ldots, 2n\}$ of $2n$ agents.}
{Does $\Pot$ admit a stable matching?}

\subsection{Properties of a preference profiles}\label{sub:profile+properties}
A preference profile~$\Pot$ may have one ore more of the following three simple properties:

\begin{definition}[Completeness preferences with ties]
A profile~$\Pot$ is \emph{complete} if for each agent~$i\in V$ it holds that $V(\succeq_i) \cup \{i\} = V$; otherwise it is \emph{incomplete}.
Profile~$\Pot$ is said to have a \emph{tie} if there is an agent~$i \in V$ 
and there are two distinct agents~$x,y \in V(\succeq_i)$ with $x \indif_i y$.
\end{definition}

Note that complete preferences without ties are exactly linear orders on~$V$. 

\begin{definition}[Narcissism]
A profile~$\Pot$ is \emph{narcissistic} if each agent~$i$ \emph{strictly} prefers itself to every other acceptable agent, i.e., for each $j\in V(\succeq_i)\setminus \{i\}$ it holds that $i \succ_i j$. 
\end{definition}

\noindent The profile given in \cref{ex:running} is complete and narcissistic, and contains one tie.

We note that having complete preferences means that any two distinct agents can be matched with each other.
Thus, a stable matching must be perfect. 
As for the narcissistic property alone, 
there is no restriction on or guarantee for the existence of a stable matching. 
We will, however, see that 
requiring some restricted preferences, such as single-peaked or single-crossing preferences as formally defined below, to be \emph{also} narcissistic makes a difference (see for instance \cref{ex:sp+nar->stablematching}).

As already discussed in \cref{sec:intro}, the single-peaked and the single-crossing properties were originally introduced and studied mainly for linear preference orders (i.e., preferences without ties).
For preferences with ties, 
a natural generalization is to think of a possible linear extension of the preferences for which the single-peaked or single-crossing property holds.
We consider this variant in our paper. 

\begin{definition}[Single-peakedness]\label{def:sp}
Let $\rhd$ be a linear order on the agent set~$V$.
An agent~$i$ with preference order~$\succeq_i$ is \emph{single-peaked with respect to~$\rhd$} if
for each three distinct agents~$x, y, z$  it holds that 
\begin{align*}
 x\rhd y \rhd z \text{ and } x\succ_i y \text{ implies } y \succeq_i z\text{.} 
\end{align*}
Accordingly, profile~$\Pot$ is \emph{single-peaked} if there exists a linear order~$\rhd$ on~$V$ such that each agent from $\Pot$ is single-peaked with respect to~$\rhd$.
We refer to~$\rhd$ as a \emph{single-peaked order} of the profile.
\end{definition}

\begin{example}\label{ex:sp+nar->stablematching}
  The profile given in \cref{ex:running} is narcissistic and single-peaked with respect to
  the linear order~$1\rhd 2 \rhd 3 \rhd 4$.
  See \cref{fig:N-SC-SP_SP} for an illustration.
  In fact, as we will see in \cref{sec:complete}, a narcissistic
  and single-peaked preference profile always admits a stable matching. 
  Recall that if agent~$3$ changes its preference order to $3\succ 1 \succ 2 \succ 4$, then the
  resulting profile does not admit any stable matching and, indeed, it is also not single-peaked anymore.
\end{example}

Just as for the single-peakedness property, the single-crossingness property also requires a natural linear order of the agents, the so-called single-crossing order.
However, unlike the single-peakedness property which assumes that the preferences of an agent~$i$ over two agents are compared by their ``distance'' to the peak along the single-peaked order,
the single-crossingness property assumes that 
the agents' preferences over each two distinct agents 
change (cross) at most once.

In fact, for preferences with ties, two natural single-crossing notions are of interest.
To describe them, we introduce two notions that can be used to partition a subset of agents
according to their preferences over two distinct agents.
Given a preference profile~$\Pot=(\succ_i)_{i \in V}$, for each two distinct agents~$x$ and $y$,
let
\begin{align*}
  B_{\Pot}(x,y)\coloneqq \{i \in V \mid \succeq_i \in \Pot \wedge x\succ_i y\}
\end{align*}
denote the subset of agents~$i$ that strictly prefer $x$ to~$y$~(i.e., $x$ is better than~$y$),
and let
\begin{align*}
  T_{\Pot}(x,y) \coloneqq \{i \in V \mid \succeq_i \in \Pot \wedge x \indif_i y\}\end{align*}
denote the subset of agents~$i$ that find $x$ and $y$ to be of equal value~(i.e., $x$ and $y$ are tied with each).
Here, symbol~$B$ means ``better than'' while symbol~$T$ means ``tied with''.

\begin{definition}[Single-crossingness]\label{def:sc}
  Let $\rhd$ be a linear order on the agent set~$V$.
A preference profile~$\Pot=(\succeq_i)_{i\in V}$ is \emph{single-crossing with respect to the linear order~$\rhd$} if 
there exists a preference profile~$\Pot'=(\succeq'_i)_{i\in V}$ without ties
such that the following two conditions hold.
\begin{enumerate}[(i)]
  \item For each agent~$i\in V$, the order~$\succeq'_i$ is a \emph{linear extension} of $\succeq_i$, i.e., the order~$\succeq'_i$ is a linear order on the acceptable set~$V(\succeq_i)$ with $\succ_i \subseteq \succ'_i$.
  \item For each two distinct agents~$x,y\in V$ it holds that $T_{\Pot'}(x,y)=\emptyset$ such that either~$B_{\Pot'}(x,y)\rhd B_{\Pot'}(y,x)$ or $B_{\Pot'}(y,x)\rhd B_{\Pot'}(x,y)$.
\end{enumerate}
Accordingly, profile~$\Pot$ is \emph{single-crossing} if there exists a linear order~$\rhd$ on the agent set~$V$ such that $\Pot$ is single-crossing with respect to~$\rhd$.
\end{definition}

\begin{example}
  The profile from \cref{ex:running} is single-crossing with respect to the linear order~$1 \rhd 2 \rhd 3 \rhd 4$.
  See \cref{fig:restricted-profiles:NSC} for an illustration.
\end{example}

\begin{figure}[t]\captionsetup[subfigure]{position=t}
   \begin{subfigure}[b]{0.6\textwidth}
   \centering 
    \begin{tikzpicture}
      \begin{axis}[
          symbolic x coords={1,2,3,4},
          xlabel shift=-1ex,
          xlabel=agents (alternatives),
          y dir=reverse,
          ylabel shift=-.5ex,
          ylabel=positions in preferences,
          xtick=data,
          ytick=data,
          x = .9cm,
          y = .7cm,
          legend style={draw=none,nodes={scale=1, transform shape},at={(1.05,0.9)},anchor=north west}
          ]
        \addplot[color=darkyellow,mark=*,line width=1.2pt] coordinates {
          (1,  1)
          (2,  2)
          (3,  3)
          (4,  4)
        };

        \addplot[color=red,mark=triangle*,line width=1.2pt] coordinates {
          (1,  3)
          (2,  1)
          (3,  2)
          (4,  4)
        };

        \addplot[color=blue,mark=square*,line width=1.2pt] coordinates {
          (1,  4)
          (2,  2)
          (3,  1)
          (4,  2)
        };

        \addplot[color=green,mark=diamond*,line width=1.2pt] coordinates {
          (1,  4)
          (2,  3)
          (3,  2)
          (4,  1)
        };
        \legend{\clR{1: $1 \succ 2 \succ 3 \succ 4$},\clB{2: $2 \succ 3 \succ 1 \succ 4$},\clG{3: $3 \succ 2 \indif 4 \succ 1$},\clO{4: $4 \succ 3 \succ 2 \succ 1$}}
      \end{axis}
    \end{tikzpicture}
    \caption{Visualization of single-peaked preferences.}\label{fig:N-SC-SP_SP}
  \end{subfigure}
~~ 
  \begin{subfigure}[b]{0.4\textwidth}
   \centering
\begin{tikzpicture}

   \matrix (scprofilet) [matrix of math nodes, row sep=10pt, column sep=2.6pt] {
    1 \colon & 1 & \succ & 2  & \succ & 3 & \succ & 4\\
    2 \colon & 2 & \succ & 3 & \succ & 1 & \succ & 4\\
    3 \colon & 3 & \succ & 2 & \indif & 4 & \succ & 1\\
    4 \colon & 4 & \succ & 3 & \succ & 2 & \succ & 1\\
    \\
  };

  \begin{scope}[on background layer]
  \draw[linemark=darkyellow] 
  (scprofilet-1-2.north) -- (scprofilet-1-2.south) -- 
  (scprofilet-2-6.north) -- (scprofilet-2-6.south) -- 
  (scprofilet-3-8.north) -- (scprofilet-3-8.south) --
  (scprofilet-4-8.north) -- (scprofilet-4-8.south);

  \draw[linemark=lightrose] 
  (scprofilet-1-4.north) -- (scprofilet-1-4.south) -- 
  (scprofilet-2-2.north) -- (scprofilet-2-2.south) -- 
  (scprofilet-3-5.north) -- (scprofilet-3-5.south) --
  (scprofilet-4-6.north) -- (scprofilet-4-6.south);

  \draw[linemark=lightblue] 
  (scprofilet-1-6.north) -- (scprofilet-1-6.south) -- 
  (scprofilet-2-4.north) -- (scprofilet-2-4.south) -- 
  (scprofilet-3-2.north) -- (scprofilet-3-2.south) --
  (scprofilet-4-4.north) -- (scprofilet-4-4.south);

  \draw[linemark=lightgreen]
  (scprofilet-1-8.north) -- (scprofilet-1-8.south) -- 
  (scprofilet-2-8.north) -- (scprofilet-2-8.south) -- 
  (scprofilet-3-5.north) -- (scprofilet-3-5.south) --
  (scprofilet-4-2.north) -- (scprofilet-4-2.south);
  \end{scope} 
\end{tikzpicture}
 \caption{Visualization of (tie-sensitive) single-crossing preferences.} \label{fig:restricted-profiles:NSC}
 \end{subfigure}
  \caption{Illustration of the profile from \cref{ex:running} which has complete preferences, contain ties, and are narcissistic, single-peaked, and (tie-sensitive) single-crossing.}\label{fig:N-SC-SP}
\end{figure}
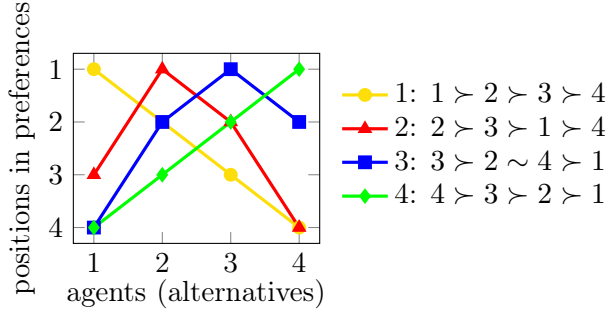
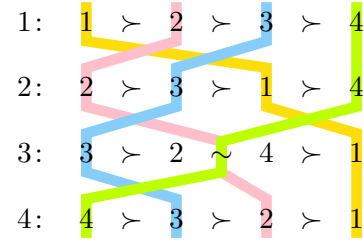

We also consider a more restricted concept of single-crossingness which requires that the agents that have ties towards a pair are ``ordered in the middle''.
\begin{definition}[\Tiescness]\label{def:tiesc}
Again, let $\rhd$ be a linear order on the agent set~$V$.
A profile~$\Pot$ on~$V$ is \emph{\tiesc{} with respect to the order~$\rhd$} if 
for each two distinct agents~$x$ and $y$ it holds that 
\begin{align*}
  \text{either } B_{\Pot}(x,y) \rhd T_{\Pot}(x,y) \rhd B_{\Pot}(y,x) 
  \text{ or } B_{\Pot}(y,x) \rhd T_{\Pot}(x,y) \rhd B_{\Pot}(x,y). 
\end{align*}
\end{definition}
\noindent See \cref{fig:restricted-profiles} for an illustration of the different types of restricted preferences for the case where the input preference orders are complete and contain no~ties.

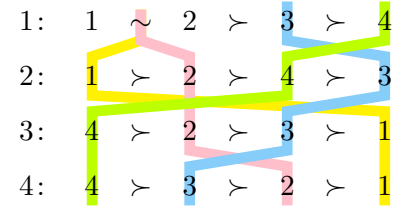
\begin{figure}[t]\captionsetup[subfigure]{position=t}
  \begin{subfigure}[b]{0.49\textwidth}
   \centering
    \begin{tikzpicture}[scale=1]
      \begin{axis}[
          symbolic x coords={1,2,3,4},
          xlabel shift=-1ex,
          xlabel=agents (alternatives),
          y dir=reverse,
          ylabel shift=-.5ex,
          ylabel=positions in preferences,
          xtick=data,
          ytick=data,
          x = 1cm,
          y = .68cm,         
          legend style={draw=none,nodes={scale=1.1, transform shape},at={(0,1.9)},anchor=north west}
          ]

        \addplot[color=darkyellow,mark=*,every mark/.append style={solid, fill=darkyellow},line width=1.2pt] coordinates {
          (1,  1)
          (2,  2)
          (3,  3)
          (4,  4)
        };

        \addplot[color=red,mark=triangle*,every mark/.append style={solid, fill=red},line width=1.2pt] coordinates {
          (1,  4)
          (2,  1)
          (3,  2)
          (4,  3)
        };

        \addplot[color=blue,mark=square*,every mark/.append style={solid, fill=blue},line width=1.2pt] coordinates {
          (1,  3)
          (2,  2)
          (3,  1)
          (4,  4)
        };
        \addplot[color=green,mark=star,every mark/.append style={solid, fill=green},line width=1.2pt] coordinates {
          (1,  4)
          (2,  3)
          (3,  2)
          (4,  1)
        };
        \legend{\clO{1: $1 \succ 2 \succ 3 \succ 4$},\clB{2: $2 \succ 3 \succ 4 \succ 1$},\clG{3: $3 \succ 2 \succ 1 \succ 4$},\clR{4: $4 \succ 3 \succ 2 \succ 1$}}
      \end{axis}
    \end{tikzpicture}
    \caption{A \SR instance with narcissistic and single-peaked preferences.
      They are \emph{not} single-crossing because of the following.
      To form a single-crossing order,
      due to pair~$\{1,4\}$, agent~$1$ must be next to agent~$3$,
      and agent~$2$ must be next to agent~$4$.
      Moreover, due to pair~$\{2,3\}$,
      agent~$1$ must be next to agent~$2$,
      and agent~$3$ must be next to agent~$4$.
All these four conditions, however, cannot be satisfied by a linear order.} \label{fig:restricted-profiles:NSP-notTSC}
  \end{subfigure}
~~  
  \begin{subfigure}[b]{0.47\textwidth}
   \centering
\begin{tikzpicture}
  \matrix (scprofile) [matrix of math nodes, row sep=0pt, column sep=0pt] {
   \Agent~1 \colon & 1 & \indif & 2  & \succ & 3 & \succ & 4\\
   \Agent~2 \colon & 1 & \succ & 2 & \succ & 4 & \succ & 3\\
    \Agent~3 \colon & 4 & \succ & 2 & \succ & 3 & \succ & 1\\
    \Agent~4 \colon & 4 & \succ & 3 & \succ & 2 & \succ & 1\\
  };
\end{tikzpicture}

\xhrulefill{gray}{1pt}\par
~\\
\begin{tikzpicture}  
  \begin{scope}[yshift=-3cm]
  \matrix (scprofilet) [matrix of math nodes, row sep=6pt, column sep=4pt] {
    1 \colon & 1 & \sim & 2  & \succ & 3 & \succ & 4\\
    2 \colon & 1 & \succ & 2 & \succ & 4 & \succ & 3\\
    3 \colon & 4 & \succ & 2 & \succ & 3 & \succ & 1\\
    4 \colon & 4 & \succ & 3 & \succ & 2 & \succ & 1\\
  };
  \end{scope}
  \begin{scope}[on background layer]
  \draw[linemark=yellow] 
  (scprofilet-1-3.north) -- (scprofilet-1-3.south) -- 
  (scprofilet-2-2.north) -- (scprofilet-2-2.south) -- 
  (scprofilet-3-8.north) -- (scprofilet-4-8.north) --
  (scprofilet-4-8.south); 

  \draw[linemark=lightrose] 
  (scprofilet-1-3.north) -- (scprofilet-1-3.south) -- 
  (scprofilet-2-4.north) -- (scprofilet-2-4.south) -- 
  (scprofilet-3-4.north) -- (scprofilet-3-4.south) --
  (scprofilet-4-6.north) -- (scprofilet-4-6.south);

  \draw[linemark=lightblue] 
  (scprofilet-1-6.north) -- (scprofilet-1-6.south) -- 
  (scprofilet-2-8.north) -- (scprofilet-2-8.south) -- 
  (scprofilet-3-6.north) -- (scprofilet-3-6.south) --
  (scprofilet-4-4.north) -- (scprofilet-4-4.south);

  \draw[linemark=lightgreen]
  (scprofilet-1-8.north) -- (scprofilet-1-8.south) -- 
  (scprofilet-2-6.north) -- (scprofilet-2-6.south) -- 
  (scprofilet-3-2.north) -- (scprofilet-3-2.south) --
  (scprofilet-4-2.north) -- (scprofilet-4-2.south);
\end{scope}
\end{tikzpicture}
    \caption{Top: A \SR instance with single-crossing preferences, with a single tie.
They are \emph{not} \tiesc since~$\{2,3,4\}$ implies that 
$1\rhd 2\rhd 3\rhd 4$ and its reverse are the only possible single-crossing orders.
However, regarding~$\{1,2\}$ the preferences are \tiesc \emph{neither} with respect to~$\rhd$ \emph{nor} with respect to its reverse.
             Bottom: A possible linear extension, showing single-crossingness.  
           } \label{fig:restricted-profiles:SC-notST}
  \end{subfigure}
  \caption{Visualization of preference profiles with different structural properties.}\label{fig:restricted-profiles}
\end{figure}

\subsection{Basic observations for the structural properties of preferences}\label{sub:basic-obser}
Incomplete preferences with ties, i.e., weak orders, as used in this work are incomparable to partial orders because partial orders are antisymmetric.
However, the asymmetric part of a weak order is indeed a partial order.
There are many slightly different concepts of single-peakedness and single-cross\-ing\-ness for partial orders~\cite{Lack2014,ElkFalLacObr2015,FiHe2016}.
It is known that detecting single-peakedness or single-crossing\-ness is NP-hard for partial orders under most of the concepts studied in the literature~\cite{BartholdiTrick1986,BaHa2011,DoiFal1994,EscLanOez2008,ElkFalSli2012,BreCheWoe2013a}.
For partial orders, our two single-crossing concepts are \emph{incomparable}.
In particular, there are incomplete preferences with ties which are single-crossing but \emph{not} \tiesc, and the converse also holds.
For complete preferences with ties which is a restricted case of partial orders, \citet{ElkFalLacObr2015} showed that \tiesc{} preference profiles are a strict subset of single-crossing preference profiles.
In the following, we extend these results by considering the case when the preferences can be incomplete.

\begin{proposition}\label[proposition]{prop:relation-sc-tsc}
  \begin{enumerate}[(1)]
    \item\label{statement:noties:tiesc=sc} For preferences without ties, \tiesc{ness} is equivalent to single-crossingness. 
    \item\label{statement:ties:sc-not=>tiesc} For preferences with ties, single-crossingness does not always imply \tiesc{ness}.
    \item\label{statement:ties:tiesc=>sc} For preferences with ties, \tiesc{ness} implies single-crossingness.
  \end{enumerate}
\end{proposition}

\begin{proof}
  Statement~\eqref{statement:noties:tiesc=sc} follows from the observation that when a profile does not have any ties, the definitions of single-crossingness (\cref{def:sc}) and \tiescness \cref{def:tiesc} coincide.
  The preference profile given in \cref{fig:restricted-profiles:SC-notST} is single-crossing but \emph{not} \tiesc{}, which shows Statement~\eqref{statement:ties:sc-not=>tiesc}.

  It remains to show Statement~\eqref{statement:ties:tiesc=>sc}. Let~$\Pot=(\succeq_i)_{i\in V}$ be a profile which is \tiesc{} with respect to some linear order~$\rhd$ on~$V$.
  We resolve the ties in the preferences given in $\Pot$ according to an arbitrary but fixed order on~$V$ and show that this extended profile fulfills the conditions given in the definition for \tiescness{} (see \cref{def:tiesc}). 
  To this end, let $\border$ be an arbitrary but fixed linear order on~$V$ and let $\Pot'=(\succeq'_i)_{i\in V}$ be a copy of the profile~$\Pot$ with $\succeq'_i=\succeq_i$, $i\in V$.
  For each agent~$i\in V$ and for each pair~$\{x,y\}\subseteq V(\succeq_i)$ of agents acceptable to $i$ such that $x\indif_i y$,
  let $x \succ'_i y$ (i.e., delete $(x,y)$ from $\succeq'_i$) if and only if $x \border y$.
  In this way, for each agent~$i$ and for each two distinct agents~$x,y \in V(\succeq_i)$ acceptable to~$i$, we have that 
  \begin{align}\label{eq:linear-extension}
    x \succ'_i y \text{ implies } x \succeq_i y,
  \end{align}
  since the preference order~$\succeq_i$ is a weak order on~$V(\succeq_i)$.
  
  First of all, we claim that $\Pot'$ fulfills the first condition given in \cref{def:sc}.
  For the sake of contradiction, suppose that there is an agent~$i$ such that $\succeq'_i$ is not a linear extension of~$\succeq_i$.
  Observe that by our construction,~$\succeq'_i$ is a complete binary relation on~$V(\succeq_i)$.
  Moreover, for each two acceptable agents~$x, y \in V(\succeq_i)$ it holds that $x\succ_i y$ implies that $x \succ'_i y$. 
  Hence, 
  the assumption that $\succeq'_i$ is not a linear extension of $\succeq_i$ means that $\succ'_i$ is not transitive, i.e.,
  there are three acceptable agents~$x,y,z\in V(\succeq_i)$ with
  \begin{align}\label{eq:assumption}
    x\succ'_i y\text{, } y\succ'_i z\text{, and }z\succ'_i x.
  \end{align}
  By Property~\eqref{eq:linear-extension}, we infer that
  $x\succeq_i y$, $y \succeq_i z$, and $z \succeq_i x$.
  By the transitivity of $\succeq_i$, these three relations imply that
  $x\sim_i y$, $y \sim_i z$, and $z \sim_i x$.
  This means that the pairwise relative orders of $x,y,z$ in $\succeq'_i$ are resolved according to the order~$\border$, i.e., $x\border y$, $y\border z$,
  and $z\border x$---a contradiction to $\border$ being transitive.

  It remains to show that $\Pot'=(\succeq'_i)_{i\in V}$ is single-crossing with respect to $\rhd$.
  Suppose towards a contradiction that $\Pot'$ is not single-crossing with respect to~$\rhd$,
  and suppose that there are three agents~$i,j,k$ in the order~$i\rhd j \rhd k$ and there are two acceptable agents~$x,y\in V(\succeq'_i)\cap V(\succeq'_j) \cap V(\succeq'_k)$ such that
  \begin{align}\label{eq:sc-assumption}
    x\succ'_i y, y\succ'_j x\text{, and }x \succ'_k y. 
  \end{align}
  By Property~\eqref{eq:linear-extension}, we infer that  
  \begin{align}\label{eq:three-agents}
    x\succeq_i y\text{, } y\succeq_j x, \text{ and } x \succeq_k y.
  \end{align}
  We distinguish between two cases for $x \succeq_j y$, in each case aiming at reaching a contradiction.

  If $x\indif_j y$, then by \eqref{eq:sc-assumption} 
  it follows that $y \border x$ and that $\neg(y\sim_i x)$ and $\neg(y \sim_k x)$.
  Thus, by \eqref{eq:three-agents}, it follows that $x \succ_i y$ and $x\succ_k y$%
  ---a contradiction to $\rhd$ being a \tiesc{} order for $\Pot$.

  If $y \succ_j x$, then since $\rhd$ is a \tiesc{} order for $\Pot$,
  by \eqref{eq:three-agents} it follows that
  $\neg(x\succ_i y)$ or $\neg(x \succ_k y)$.
  This implies that $x\indif_i y$ or $x \indif_k y$.
  If $x \indif_i y$, then $y \succ_k x$---a contradiction to \eqref{eq:three-agents}.
  If $x \indif_k y$, then $y \succ_i x$---again a contradiction to \eqref{eq:three-agents}.
%
%
\end{proof}

For incomplete preferences with ties, 
\citet{Lack2014} showed that detecting single-peakedness is NP-complete.
For complete preferences with ties, 
while \citet{ElkFalLacObr2015} showed that detecting single-crossingness is NP-complete,
\citet{FitzSPW15} and \citet{ElkFalLacObr2015} provided polynomial-time algorithms for detecting single-peakedness and ties-sensitive single-crossingness, respectively. 
All these known hardness results seem to hold only when the preferences have ties.
However, we observe that the hardness reduction for Corollary~6 of \citet{ElkFalLacObr2015} indeed can be adapted to show NP-hardness for deciding whether an incomplete preference profile without ties is single-peaked (resp.\ single-crossing).
The crucial differences are that they allow ties and that the agents and the alternatives are different while we do not allow ties and our agent set is the same as the set of alternatives. 
For the sake of completeness, we show this adapted reduction. 

\begin{proposition}\label{prop:tiesless-incomp-sp-sc-np-c}
  Deciding whether a preference profile with incomplete preferences and without ties is single-crossing (or equivalently \tiesc) or single-peaked is NP-complete. 
\end{proposition}

\begin{proof}
  The three decision problems are clearly contained in NP as given a linear order~$\rhd$ on the agent set one can verify in polynomial-time whether the given preference profile is single-crossing, \tiesc, or single-peaked with respect to the order~$\rhd$; note that the preferences do not contain ties.

  To show NP-hardness, as \citet{ElkFalLacObr2015}, we reduce from the NP-hard \textsc{Betweenness} problem~\cite{Opa79}:
  \decprob{Betweenness}
  {Given a universe $U= \{u_1,\ldots,u_n\}$ and a set~$T=\{t_1,\ldots,t_m\}$ of ordered triples over~$U$.}
  {Is there a \emph{betweenness order~$\border$}, that is, a total linear order over~$U$ such that
   for each triple~$(x,y,z)$ from~$T$ it holds that
   either $x \border y \border z$ or $z \border y \border x$?}

 \noindent
 The reduction for the single-peaked case is quite simple while
 the reduction for the single-crossing case is similar to the one used by \citet[Corollary 6]{ElkFalLacObr2015}.

 \smallskip
 
 \noindent \textbf{The single-peakedness case.} Given an instance~$I=(U,T)$ of \textsc{Betweenness},
 where $U=\{u_1,\ldots,u_n\}$ and $T=\{t_1,\ldots,t_m\}$, 
 we construct a preference profile~$\Pot=(\succ_x)_{x\in V}$ for an agent set~$V$, whose preferences contain no ties but may be incomplete; without loss of generality we assume that the elements in each ordered triple~$(x,y,z)\in T$ are pairwise distinct.
 The agent set~$V$ has two types of agents, which sum up to $n+2m$~agents.
 First, for each element~$u_i$ add to $V$ an \emph{element} agent with the same name.
 Second,  for each triple~$t_j=(x,y,z)$ from~$T$ (with~$x,y,z \in U$) add to $V$ two agents~$a_j$ and~$a'_j$.
 We will construct the preference orders of~$a_j$ and~$a'_j$ to ensure that the second element~$y$ in triple~$t_j$ is ordered between the other two elements~$x$ and $y$ in a betweeness order.

 The preference orders of the agents are constructed as follows:
 \begin{itemize}
   \item Each element agent~$u_i \in U$ only finds those agents acceptable that ``contain''~$u_i$, and has a linear order on these acceptable agents that are consistent with the following order 
 $$L\coloneqq a_1 \succ \cdots \succ a_m \succ a'_1 \succ \cdots \succ a'_m.$$
 Formally, the preference order of agent~$u_i$ is
 $$\succ_{u_i}\colon [A(u_i)] \succ [A'(u_i)],$$
 where $A(u_i)=\{a_j\mid u_i \in t_j\}$ (resp.\ $A'(u_i)=\{a'_j\mid u_i \in t_j\}$)
 and $[A(u_i)] \succ [A'(u_i)]$ is a linear order on $A(u_i) \cup A'(u_i)$ that respects~$L$.
 
 \item  For each triple~$t_j=(x,y,z) \in T$, the preference orders of $a_j$ and $a'_j$ are
 \begin{align*}
  \succ_{a_j}\colon y \succ x \succ z \text{ and } \succ_{a'_j}\colon y \succ z \succ x.  
 \end{align*}
\end{itemize}

 We show that $\Pot$ is single-peaked if and only if the given instance is a yes-instance.
 On the one hand, every betweenness order~$\border$ for~$I=(U,T)$ can be extended to a single-peaked order for~$\Pot$ by appending to~$\border$ the order~$L$. 
 On the other hand, note that the preference orders of the agents of the second type require a single-peaked order to place element~$y$ in between $x$ and $z$.
 Thus, a single-peaked order~$\rhd$ for $\Pot$ restricted to~$U$ is also a betweenness order for~$(U,T)$.
 
 \smallskip
 \noindent \textbf{The single-crossingness and \tiescness{} case.}
 Since for preferences without ties, single-crossingness is equivalent to \tiescness~(see \cref{prop:relation-sc-tsc}\eqref{statement:noties:tiesc=sc}) we only focus on the single-crossingness case.
Given an instance~$I=(U,T)$ of \textsc{Betweenness},
where $U=\{u_1,\ldots,u_n\}$ and $T=\{t_1,\ldots,t_m\}$, 
we construct a preference profile~$\Pot=(\succ_x)_{x\in V}$ for an agent set~$V$, whose preferences contain no ties but may be incomplete; without loss of generality we assume that the elements in each ordered triple~$(x,y,z)\in T$ are pairwise distinct.

The agent set~$V$ has two types of agents, which sum up to $n+3m$~agents.
 First, for each element~$u_i$ add to $V$ an \emph{element} agent with the same name.
 Second,  for each triple~$t_j=(x,y,z)$ from~$T$ (with~$x,y,z \in U$) add to $V$ three agents~$a_j,b_j, c_j$.
 Again, we will construct the preference orders of~$a_j, b_j$, and $c_j$ to ensure that the second element~$y$ is ordered between the other two elements~$x$ and $z$. 
 Before we describe how to construct the preference orders, let us describe the acceptable set of each agent:
 For each triple~$t_j=(x,y,z)$ from~$T$ (with~$x,y,z \in U$) let
 $$V(\succ_{a_j})=V(\succ_{b_j})=V(\succ_{c_j})\coloneqq \{x,y,z\}.$$
 Symmetrically, for each element agent~$u\in U$, let $V(u) \coloneqq \{a_j,b_j,c_j \mid u \in t_j\}$.
 For ease of notation, for each triple~$t_j\in T$, let $T_j \coloneqq \{a_j,b_j,c_j\}$.

 The preference orders of the agents are constructed as follows:
 \begin{itemize}
   \item For each triple~$t_j =(x,y,z)$ from~$T$, the preference orders of~agents $a_j$, $b_j$, and $c_j$ are the same linear order on $\{x,y,z\}$, which is ascending on the indices of $x,y,z$ in~$U$.
   \item For each element agent~$u\in U$, 
   construct the preference order~$\succ_u$ of $u$ such that for each two triples~$t_j$ and $t_{j'}$ with $j < j'$ that contain the element~$u$ it holds that
   \[ \succ_u\colon T_j \succ T_{j'}.\]
   The specific order on the triple agents in $T_j$ (resp.\ $T_{j'}$) depends on the position of $u$ in the triple~$t_j$.
   To this end, let
   \begin{align*}
     L_j^x & \coloneqq a_j \succ b_j \succ c_j,\\
     L_j^y & \coloneqq b_j \succ a_j \succ c_j,\\
     L_j^z & \coloneqq c_j \succ b_j \succ a_j.
   \end{align*}
   To make the preference order of $u$ a linear order, for each triple~$t_j\in T$ with $t_j=(x,y,z)$ that contains the element~$u$,
   let $\succ_u$ obey the order~$L_j^{k}$ if and only if $u=k$.
 \end{itemize}

 We show that $\Pot$ is single-crossing if and only if the given instance is a yes-instance.
 On the one hand, a single-crossing order~$\rhd$ for the constructed profile restricted to~$U$ is a betweenness
 order for~$(U,T)$: For every triple~$t_j=(x,y,z)$ from~$T$ the preference orders of
 $x$, $y$, and $z$ restricted to $T_j$ imply that either
 $x \rhd y \rhd z$ or $z \rhd y \rhd x$.
 On the other hand, a betweenness order~$\border$ for~$(U,S)$ can be extended to a single-crossing order for the
 constructed profile by appending the order~$\border a_1 \border b_1 \border c_1 \border \cdots a_m \border b_m \border c_m$
 to its end. 
\end{proof}

\citet{ST06} showed that for complete preferences without ties, 
 narcissistic and single-crossing preferences are also single-peaked.
We strengthen this result by showing that the relation also holds when ties are allowed.
Further, we note that although \citet{BaMo2011} also considered complete preferences with ties, their single-crossingness for the case with ties only resembles our \tiesc definition, 
which is a strict subset of the single-crossingness defined in \cref{def:sc} (also see \cref{prop:relation-sc-tsc}\eqref{statement:ties:sc-not=>tiesc}--\eqref{statement:ties:tiesc=>sc}).

\begin{proposition}\label[proposition]{thm:complete:nsc->sp}
  If a complete, even with ties, and narcissistic preference profile $\Pot$ admits a single-crossing order~$\rhd$, 
  then this order~$\rhd$ is also a single-peaked order. 
\end{proposition}

\begin{proof}
  Suppose for the sake of contradiction that $\rhd$ with $a_1\rhd a_2 \rhd \dots \rhd a_{2n}$ is not single-peaked.
  This means that there exists an agent~$a_i$ that is not single-peaked with respect to~$\rhd$,
  and there are three agents~$a_j,a_k,a_\ell$ with
  $a_j \rhd a_k \rhd a_{\ell}$ such that $a_j\succ_{a_i} a_k$ and $a_\ell \succ_{a_i} a_k$.
  Together with the property of being narcissistic, the following holds:
\begin{align*}
     \text{agent } a_i\colon  & a_i \succ_{a_i} a_{j} \succ_{a_i} a_{k} \text{ and } a_i \succ_{a_i} a_{\ell} \succ_{a_i} a_{k}, &\quad
    \text{agent } a_j\colon  & a_j \succ_{a_j}  a_k,\\
    \text{agent } a_k\colon  & a_k \succ_{a_k} a_j \text{ and } a_k \succ_{a_k} a_{\ell}, &\quad
    \text{agent } a_\ell\colon  &a_\ell \succ_{a_\ell} a_k\text{.}
\end{align*}

  \noindent On the one hand, the agents' preferences over the pair~$\{a_j, a_k\}$ implies that $i < k$. 
  On the other hand, the pair~$\{a_k, a_\ell\}$ implies that $i > k$---a contradiction. 
\end{proof}

The profile shown in \cref{fig:N-SC-SP} is narcissistic and single-crossing
with respect to the order $1\rhd 2 \rhd 3 \rhd 4$ and it is also single-peaked with respect to the same order~$\rhd$.

There is no direct relation between single-peakedness and single-crossingness, even if the profile is complete and do not contain ties~\cite{BreCheWoe2016}.

\begin{proposition}
  Even for complete preference profiles without ties,
  single-peaked\-ness does not imply single-crossingness.
  Neither does single-crossingness imply single-peakedness.
\end{proposition}

\begin{proof}
  The statements are indicated in~\cite[Figure 1]{BreCheWoe2016}. For the sake of completeness, however, we show the statements through the following two concrete profiles.

Consider the following profile with four agents~$1,2,3,4$:
\begin{alignat*}{5}
  \text{Profile~}&&\Pot_1\colon &  \agent~1\colon &&1 \succ 2 \succ 3 \succ 4,\quad &~~&
    \agent~2\colon &&2 \succ 1 \succ 3 \succ 4, \\
 & &&  \agent~3\colon &&1 \succ 2 \succ 4 \succ 3, &&
    \agent~4\colon &&2 \succ 1 \succ 4 \succ 3.
  \end{alignat*}

  One can verify that profile~$\Pot_1$ has complete preferences without ties,
  and is single-peaked with respect to the linear order~$3 \rhd 1 \rhd 2 \rhd 4$.
  However, it is not single-crossing because agents~$1$, $2$, $3$, and~$4$ form a forbidden configuration of single-crossingness~\cite{BreCheWoe2013a}.
   



Consider the following profile with four agents~$1,2,3,4$:
\begin{alignat*}{4}
  \text{Profile~}\Pot_2\colon\quad &  \agent~1\colon &&1 \succ 2 \succ 3 \succ 4,\quad &~~&
    \agent~2\colon &&1 \succ 2 \succ 3 \succ 4, \\
  &  \agent~3\colon &&1 \succ 4 \succ 3 \succ 2, &&
    \agent~4\colon &&1 \succ 4 \succ 3 \succ 2.
 \end{alignat*}
Since there are only two different preference orders,
profile~$\Pot_2$ is obviously single-crossing, for instance, with respect to the linear order~$1\rhd2\rhd3\rhd4$.
However, it is not single-peaked since agent~$1$ and agent~$3$ form a forbidden configuration of single-peakedness~\cite{BaHa2011}.
\end{proof}

\section{Complete preferences}\label{sec:complete}
In this section, we analyze the computational complexity of \SR{} for the case
when the input preference profiles have complete and structured preferences.
In particular, we show that the 
NP-hard \SR{} problem with ties allowed becomes
polynomial-time solvable when the preferences are narcissistic and either single-crossing or single-peaked.

For the case of complete, narcissistic, and single-peaked preferences without ties,
\citet{BartholdiTrick1986} showed that there is always a unique stable matching and provided an algorithm to find it.
Their algorithm is based on the following two facts (referred to as \cref{prop:nsp->twomost,prop:stablematching-augmented})
that are related to the concept of most acceptable agents.
We show that the facts transfer to the case with ties.

\begin{proposition}\label[proposition]{prop:nsp->twomost}
  If the given preference profile~$\Pot$ is complete (even with ties), narcissistic, and single-peaked,
  then there are two distinct agents~$i,j$ that are each other's most acceptable agents. 
\end{proposition}

\begin{proof}
  The statement for complete, narcissistic, and single-peaked preferences without ties was shown by \citet{BartholdiTrick1986}.
  It turns out that this also holds for the case when ties are allowed.
  Let $V$ be the set of all $2n$~agents and consider a single-peaked order~$\rhd$ of the agents~$V$ with $x_1\rhd x_2\rhd \dots \rhd x_n$.
  For each agent~$x\in V$, let $M_x$ be the set of all most acceptable agents of~$x$.
  Towards a contradiction, suppose that each two distinct agents~$x$ and $y$ 
have $x\notin M_y$ or $y\notin M_x$.
  By the narcissistic property and single-peakedness, each~$M_x\cup \{x\}$ forms an interval in $\rhd$.
  This implies that the first agent~$x_1$ and the last agent~$x_n$ in the order~$\rhd$ have $x_2\in M_{x_1}$ and $x_{n-1} \in M_{x_n}$.
  By our assumption ($x\notin M_y$ or $y\notin M_x$), however,
  $x_2 \in M_{x_1}$ implies that for each $i\in \{2,\dots, n\}$
  the following holds: $x_{i-1}\notin M_{x}$---a contradiction to $x_{n-1}\in M_{x_n}$.  
\end{proof}

By the stability definition, we obtain the following for complete preferences.
\begin{proposition}\label[proposition]{prop:stablematching-augmented}
  Let $\Pot$ be a preference profile with complete preferences and let $M$ be a stable matching for $\Pot$.
  Let $\Pot'$ be a preference profile resulting from $\Pot$ by adding two agents~$x,y$ who consider each other as most acceptable (and their preferences over other agents and the preferences of other agents over $x,y$ are arbitrary but fixed, respectively).
  Then, matching~$M\cup \{\match{}{x}{y}\}$ is stable for~$\Pot'$.
\end{proposition}
\begin{proof}
  Suppose for the sake of contradiction that $M\cup \{\match{}{x}{y}\}$ is not stable for~$\Pot'$.
  This means that $\Pot'$ has a blocking pair~$\{u,w\}\notin M$.
  Obviously, $|\{u,w\}\cap \{x,y\}| = 1$ as otherwise 
  $\{u,w\}$ would also be a blocking pair for~$\Pot$.
  Assume without loss of generality that $u = x$.
  Then, by the definition of blocking pairs, it must hold that $w \succ_x y$---a contradiction to $y$ being one of the most acceptable agents of~$x$. 
\end{proof}

\begin{algorithm}[t!]
   \footnotesize
  \SetKwInput{KwI}{Input}
  \SetKwFunction{BTAlg}{BT-Agorithm}
   \SetKwBlock{Block}
   \SetAlCapFnt{\footnotesize}
     $M\leftarrow \emptyset$\;
     
   \While{$\Pot\neq\emptyset$}{
     Find two agents~$x,y$ in $\Pot$ that consider each other as most acceptable\;
     Delete $x$ and $y$ from profile~$\Pot$\;
     $M\leftarrow M\cup \{\match{}{x}{y}\}$\;
   }
   \Return $M$\;
\caption{The algorithm of \citet{BartholdiTrick1986} for computing a stable matching with input~$\Pot$ being complete, narcissistic, and single-peaked, without ties.}\label{alg:BT}
 \end{algorithm}

Utilizing restricted versions of \cref{prop:nsp->twomost,prop:stablematching-augmented},
\citet{BartholdiTrick1986} derived a greedy algorithm to construct a \emph{unique} stable matching when the preferences are linear orders (i.e., complete and without ties) and are narcissistic and single-peaked (see \cref{alg:BT}). 
For $2n$~agents they claimed that their algorithm runs in $O(n)$ time by observing that, after adding a pair to the solution matching, at most two other agents need to update their most acceptable agent.
However, since updating the most acceptable agent may take~$O(n)$~time,
it is not clear how to implement their algorithm to achieve $O(n)$ running time.
Indeed, by generalizing their example~\cite[Section 3]{BartholdiTrick1986} to the case with $2n$ agents,
we show the following: Assuming, as input, an integer vector for each agent is given, which lists the agents in order of its preferences, their algorithm may take $O(n^2)$~time, even if the single-peaked order is known.

\begin{example}\label{ex:BT-n^2}
  We consider a complete, narcissistic, single-peaked, and single-crossing preference profile~$\Pot_n$ without ties, for $2n$~agents.
  The preference orders will be pasted together from the following ``preference pieces''~$X_i$ and $Y_i$ with $i \in [n]$:
  \begin{align*}
    X_i & \coloneqq  i+1 \succ i+2 \succ \cdots \succ 2n-i, \text{ and }\\
    Y_i & \coloneqq  i \succ 2n+1-i.
  \end{align*}
  Note that for every $i\in [n]$, the piece~$X_i$ covers contiguous interval of~$2n-2i$~alternatives.
  Together with the pieces~$Y_1$, $Y_2$, $\ldots$, $Y_{i-1}$, all but two alternatives are covered: $i$ and $2n+1-i$ are the only alternatives that are not covered.

  Now, let us define the preference orders of the agents in the profile.
  For each $i\in [n]$ the two agents~$i$ and $2n+1-i$ have ``almost'' opposite preferences:

  \smallskip
 \noindent \begin{tabular}{l@{\;}c@{\;}c@{\;}c@{\;}c@{\;}c@{\;}c@{\;}c@{\;}c@{\;}c@{\;}c@{\;}c@{\;}c@{\;}c@{\;}c@{\;}c@{\;}c@{\;}c@{\;}c}
    $\agent~i \colon$ & $i$ & $\succ$ &$X_i$ & $\succ$ & $2n+1-i$ &  $\succ$ & $Y_{i-1}$ & $\succ$ & $Y_{i-2}$ & $\succ$ & $\cdots$ & $\succ$ & $Y_1$,\\
    $\agent~{2n+1-i} \colon$ & $2n+1-i$ & $\succ$ & $\overleftarrow{X_i}$ & $\succ$ & $i$ & $\succ$ & $\overleftarrow{Y_{i-1}}$ & $\succ$ & $\overleftarrow{Y_{i-2}}$ & $\succ$ & $\cdots$ & $\succ$ & $\overleftarrow{Y_1}$,\\
  \end{tabular}
  
\noindent  where $\overleftarrow{X_i}$ and $\overleftarrow{Y_j}$ ($j\in [i-1]$) denote the reverse preference orders of ${X_i}$ and~${Y_j}$, respectively.

Observe that, for $i = 1$, the preference pieces~$Y_{i-1}\succ Y_{i-2} \succ \ldots \succ Y_1$ and~$\overleftarrow{Y_{i-1}}\succ \overleftarrow{Y_{i-2}} \succ \ldots \succ \overleftarrow{Y_1}$ are empty.
Thus, the first agent and the last agent have reverse preference orders~$1\succ 2 \succ \cdots \succ 2n$ and $2n\succ 2n-1 \succ \cdots \succ 1$, respectively.
  
  For $n=3$, the corresponding preference profile~$\Pot_3$ looks as follows:
  \begin{quote}
    \begin{tabular}{@{\;\quad}l@{\;}c@{\;}c@{\;}c@{\;}c@{\;}c@{\;}c@{\;}c@{\;}c@{\;}c@{\;}c@{\;}c@{\;}c@{\;}c@{\;}c@{\;}c@{\;}c@{\;}c@{\;}c}
    $\agent~1\colon$ & $1$ & $\succ$ & $2$ & $\succ$ & $3$ & $\succ$ & $4$ & $\succ$ & $5$ & $\succ$ & $6$,\\
    $\agent~2\colon$ & $2$ & $\succ$ & $3$ & $\succ$ & $4$ & $\succ$ & $5$ & $\succ$ & $1$ & $\succ$ & $6$,\\
    $\agent~3\colon$ & $3$ & $\succ$ & $4$ & $\succ$ & $2$ & $\succ$ & $5$ & $\succ$ & $1$ & $\succ$ & $6$,\\
    $\agent~4\colon$ & $4$ & $\succ$ & $3$ & $\succ$ & $5$ & $\succ$ & $2$ & $\succ$ & $6$ & $\succ$ & $1$,\\
    $\agent~5\colon$ & $5$ & $\succ$ & $4$ & $\succ$ & $3$ & $\succ$ & $2$ & $\succ$ & $6$ & $\succ$ & $1$,\\
    $\agent~6\colon$ & $6$ & $\succ$ & $5$ & $\succ$ & $4$ & $\succ$ & $3$ & $\succ$ & $2$ & $\succ$ & $1$.
  \end{tabular}
\end{quote}
  
\paragraph{Single-peakedness and single-crossingness.} Profile~$\Pot_n$ is single-peaked and single-crossing with respect to the canonical order~$1\rhd 2 \rhd \ldots \rhd 2n$.
 To see why it is single-crossing with respect to this order, we observe the following, which covers all possible pairs of agents:
 \begin{itemize}
   \item For each agent~$i \in [n-1]$ and each agent~$j \in [2n-i]\setminus [i]$, all agents from $[i]$ prefer $i$ to~$j$ while all agents from $[2n]\setminus [i]$ prefer $j$ to~$i$.
   \item For each agent~$k \in [2n]\setminus [n+1]$ and each agent $j\in [k-1]\setminus [2n+1-k]$, 
     all agents from $[k-1]$ prefer $j$ to $k$ while all agents from $[2n]\setminus [k-1]$ prefer $k$ to~$j$.
   \item For each agent~$z \in [n]$, all agents from $[n]$ prefer $z$ to $2n+1-z$ while all agents from $[2n]\setminus [n]$ prefer~$2n+1-z$ to~$z$.
 \end{itemize}
 This profile contains a unique stable matching~$M$ with
 \begin{align*}
   M = \{\{i,2n+1-i\} \mid i \in [n]\}.
 \end{align*}

 \paragraph{Running time of Bartholdi III and Tricks' algorithm~(see \cref{alg:BT}).}
 First, observe that due to the preference orders of agents~$1$ and~$2n$, the canonical order~$\rhd$ as described above is the unique single-peaked order (up to reversal).
 Hence, let us assume that the preference orders of~$\Pot_n$ and the order~$\rhd$ are given as input.
  Applying the algorithm of \citet{BartholdiTrick1986}~(see \cref{alg:BT}), in the first round,
  each agent~$i\in [n]$ regards its ``successor agent''~$i+1$ (along~$\rhd$) as the most acceptable agent while each agent~$k\in [2n]\setminus [n]$ regards its ``predecessor agent''~$k-1$ (along~$\rhd$) as the most acceptable agent.
  Thus, in this round, agents~$n$ and $n+1$ are the only agents that regard each other as most acceptable.
  We find them in $O(n)$~time and add $\{n,n+1\}$ to~$M$.

  In the second round, agents~$i\coloneqq n-1$ and $j\coloneqq n+2$ have to update their most acceptable agents since their currently most acceptable ones are matched with each other.  
  But the updating time of each agent depends on the number of pairs that are already in the solution~$M$ no matter what $\rhd$ looks like since in general the ``successor'' or ``predecessor'' of $i$ or $j$ might already have been matched in some previous round.
  For agents~$n-1$ and $n+2$, we need to traverse to the third agent in their preference orders to find one which are unmatched.
  In this case, $n-1$ and $n+2$ find each other most acceptable after $\{n,n+1\}$ is added to~$M$.
  
  For our profile, in round~$z$, starting from $z=2$,
  we need to update the most acceptable agents of agents~$n-z+1$ and $n+z$ since the pair~$\{n-z+2, n+z-1\}$ is added to~$M$ in round~$z-1$.
  To update, we need to go through the preference orders of~$n-z+1$ and~$n+z$ until we find a next most acceptable and not-yet-matched agent.
  This takes $O(z)$~time because both agents~$n-z+1$ and $n+z$ have not updated their most acceptable agents but prefer all matched agents to their next most acceptable and not-yet-matched one.
  Summarizing, we need $O(n+2+4+\ldots+2n)=O(n^2)$~time to construct a unique stable matching, if we assume that the input contains the preference orders of all agents (implemented as integer vectors) and a single-peaked order of the agents.
\end{example}

For the case with ties, 
we will show that \cref{alg:BT} also works for narcissistic, single-peaked profiles.
In particular, there is always a stable matching, albeit perhaps not unique. 

\begin{proposition}\label{prop:>1-SM}
  A complete, narcissistic, single-peaked, and (tie-sensitive) single-crossing preference profile with ties may admit more than one stable matching.
\end{proposition}

\begin{proof}
  To show the statement, let us consider the following profile:
  \begin{profile}{rccccccc}
   \agent~1 \colon & 1 & \succ & 2 & \indif & 3 & \indif & 4,\\
   \agent~2 \colon & 2 & \succ & 1 & \indif & 3 & \indif & 4,\\
   \agent~3 \colon & 3 & \succ & 1 & \indif & 2 & \indif & 4,\\
   \agent~4 \colon & 4 & \succ & 1 & \indif & 2 & \indif & 3.\\
 \end{profile}

 Since each agent regards every other agent as equally good, any perfect matching is a stable matching.
 Since the profile is complete and has four agents, there are $\binom{4}{2}/2=3$ perfect matchings which are all stable.
 It is also straight-forward to verify that the profile is narcissistic, single-peaked, and (tie-sensitive) single-crossing with respect to the order~$1\rhd 2 \rhd 3 \rhd 4$.
\end{proof}

 Next, we show that the idea behind Bartholdi III and Trick's algorithm~(see \cref{alg:BT}) also works for the case with ties.
 
\begin{theorem}\label[theorem]{thm:complete-ties-nsp-qudratic}
  If a preference profile with $2n$~agents is complete, possibly with ties, narcissistic and single-peaked,
  then it always admits a stable matching,
  which can be found 
  in $O(n^2)$ time.
\end{theorem}

\begin{proof}
  To show that such a profile~$\Pot$ always admits a stable matching, we show that on input~$\Pot$ \cref{alg:BT} always returns a matching of $\Pot$ which is stable.
  Indeed, the latter follows directly from \cref{prop:nsp->twomost,prop:stablematching-augmented} 
  and the narcissistic and single-peaked property is preserved after deleting any agent. 

  As for the running time, 
  there are $n$~rounds to build up~$M$, and 
  in each round we find two distinct agents~$x$ and $y$ whose most acceptable agent sets~$M_x$ and $M_y$ (in the updated profile) include $x$ respectively~$y$:
  \begin{align}\label{cond:good-pair}
    x\in M_y \text{ and } y\in M_x.
  \end{align}
  Note that \cref{prop:nsp->twomost} implies that two such agents exist.
  One way to search for two such agents is to go through all pairs~$\{x,y\}$ and check whether Condition~\eqref{cond:good-pair} holds.
  This leads to an algorithm with running time~$O(n^3)$.
  To obtain an $O(n^2)$-time algorithm, however, we observe that we only need to check an agent pair~$\{x,y\}$ upon Condition~\eqref{cond:good-pair}
  when the most acceptable agents of one of $x$ and $y$, say $x$, are all matched (in previous rounds)
  such that $y$ is the next less preferred agent by~$x$.
  This guarantees that we have to check each pair of agents at most twice, once for $x$ and once for $y$.
  In total, we can find a stable matching in $O(n^2)$ time.

  We describe this idea in \cref{alg:BText} (see \CcomputeSM()), which uses only simple data structures such as integer matrices and vectors.  
  As input, we assume that for each agent~$x$ there is an integer vector~$L_x$ that lists the agents in order of the preferences of~$x$, with ties resolved in an arbitrary way. 
  Moreover, we assume to have an $n \times n$ integer rank matrix $R$ where~$R[x,y]$ contains the rank of agent~$x$ in the preference list of agent~$y$. 
  Herein, the \emph{rank} of agent~$y$ in agent~$x$'s preference list is defined as the number of
  agents that are strictly preferred to~$y$ by $x$.
  For example, if the preference list of agent~$1$ is $1 \succ 2 \succ 3 \indif 4 \succ 5 \indif 6 \succ 7$,
  then we may have $L_1=[1,2,3,4,5,6,7]$
  and have $R[1,1] = 0$, $R[1,2]=1$, $R[1,3]=2$, $R[1,4]=2$, $R[1,5]=4$, $R[1,6]=4$, and $R[1,7]=6$.

 \begin{algorithm}[t!]

   \SetKwInput{KwI}{Input}
   \SetKwInput{KwO}{Output}
 
   \SetKwProg{Fn}{}{}{end}
   \SetKwBlock{Block}

   \SetAlCapFnt{\footnotesize}
   \KwI{A 
    profile~$\Pot$ for a set~$V=\{1,2,\ldots, 2n\}$ of agents with complete, narcissistic, and single-peaked preferences
  }
  \KwO{A stable matching~$M$ for~$\Pot$}

  
  \Comment{%
    \textbf{Input data structures}:
  \begin{itemize}[-]
    \item a preference list~$L_x$ for each agent~$x\in V$, and
    \item an $n \times n$ integer rank matrix~$R$, with $R[x,y]=$ rank of agent~$y$ in agent~$x$'s preferences.
  \end{itemize}
  $\rhd$ \textbf{Introduced data structures:}
  \begin{itemize}[-]
    \item an $n \times n$ binary matrix~$A$, 
    \item[] ~with $A[x,y]=1$ iff agent~$y$ is currently one of the most acceptable agents of $x$;
   \item  a size-$n$ integer vector~$B$, 
   \item[] ~with $B[x]$ storing the number of most acceptable agents of~$x$ that are not yet matched;
   \item
   a queue $Q$~of `possible matching events', containing unordered pairs of agents;
   \item a size-$n$ binary vector~$D$, 
   with $D[x]=1$ iff agent~$x$ has already been matched.
   \item a size-$n$ integer vector~$\pointer\in [n+1]^n$,
   \item[] ~with $L_z[\pointer[z]]$ being the `first' agent not from the most acceptable set of $z$.
 \end{itemize}
}

  \Fn{\computeStableMatching{$(L_x)_{x\in V}, R$}}{
    $M\leftarrow \emptyset$; $A\leftarrow 0^{n\times n}$; $B \leftarrow 0^{n}$; $Q \leftarrow \emptyset$; $D \leftarrow 0^{n}$\;

   \ForEach{agent $z\in V$}{
     $\pointer[z] \leftarrow 2$\Comment*{start with the first agent~$\neq z$ from agent~$z$'s preferences}
       \CinsertMP($z$)\;  \label{insert:first-calls}  
   }
   \While{$Q\neq \emptyset$}{
    Pop the next possible matching event $\{x,y\}$ from~$Q$\;
    \If{$D[x]=0$ \AND $D[y]=0$\label{detect}}{
      \CsmartD($x$)\;
      
      \CsmartD($y$)\;

      $M\leftarrow M\cup \{\match{}{x}{y}\}$\;
    }
   }
   \Return $M$\;
   }

   \Fn{\insertMostPreferred{$z$}}{     
     \If{$\pointer[z] \le n$}{%
       $y' \leftarrow L_z[\pointer[z]-1]$\Comment*{initialize local variable storing last inserted agent}
       
       \Repeat{($\pointer[z]\ge n+1$) \OR ($B[z]>0$ \AND $R[z,y']<R[z,L_z[\pointer[z]]]$)}
       {
         \If{$D[L_z[\pointer[z]]]=0$\label{insert:unmatched}}{
           $y' \leftarrow L_z[\pointer[z]]$\Comment*{next most acceptable agent~$\neq z$}
            $A[z,y']\leftarrow 1$\; \label{insert:set-one} 
           \If{$A[y',z]=1$\label{insert:push-if}}{Push $\{z,y'\}$ to $Q$\Comment*{$z$ and $y'$ are currently each other's most acceptable agents}\label{insert:push-do}}
           $B[z]\leftarrow B[z]+1$\;
         }
         $\pointer[z] \leftarrow \pointer[z]+1$\Comment*{look for next most acceptable agent} \label{insert:next}
       }
     }
   }
   
   \Fn{\smartDelete{$x$}}{
   $D[x]\leftarrow 1$\;

   \lForEach{agent~$z\in V$ with $A[x,z]=1$}{
    $A[x,z]\leftarrow 0$\label{delete:set-zero1}
    }
   
   \ForEach{agent~$z\in V$ with $A[z,x]=1$}{
   $A[z,x]\leftarrow 0$\;\label{delete:set-zero2}
   $B[z]\leftarrow B[z]-1$\;
    \lIf{$B[z]=0$}{\CinsertMP($z$)}
  }
  }
 \caption{Computing a stable matching with input~$\Pot$ being complete, narcissistic, single-peaked,
 and possibly with ties in $O(n^2)$~time.}
 \label{alg:BText}
 \end{algorithm}

 Before we show the correctness of \cref{alg:BText}, we first claim the following.
  Each call of \CinsertMP($z$)
  \begin{compactenum}[(a)]
    \item\label{inv:z-not-deleted} assumes $D[z]=0$, 
    \item\label{inv:find-next} finds all unmatched agents~$y'$ which are tied with agent~$L_z[\pointer_z]$ and sets $A[z,y']=1$~(Lines~\eqref{insert:unmatched}--\eqref{insert:next}), and
    \item\label{inv:push-next} pushes a possible matching event~$\{z,y'\}$ to $Q$ if, additionally, $y'$ also finds $z$ most acceptable~(Lines~\eqref{insert:push-if}--\eqref{insert:push-do}).
  \end{compactenum}
  For the first calls of \CinsertMP($z$) in Line~\eqref{insert:first-calls}, the above invariant is straight-forward to verify:
  The matrix~$A$ and the vector~$D$ are initialized with all-zero entries, and $\pointer[z]$ points to a position in the preference list~$L_z$ of~$z$ such that
  \CinsertMP($z$) goes through the preference list~$L_z$ of~$z$, using pointer~$\pointer[z]$, 
  and finds all most acceptable agents of~$z$, updating the row~$A[z]$,
  until $\pointer[z]$ points to a position which contains an agent that is strictly less preferred.

  For each following execution of \CinsertMP($z$), observe that the procedure is only called when~$D[z]=0$ and when $B[z]=0$.
  Since $B[z]$ is increased whenever some entry~$A[z,x]$ is set from zero to one
  and decreased whenever some entry~$A[z,x]$ is set from one to zero,
  this implies that $A[z,x]=0$ holds for each agent~$x\in V$ when the procedure is called.
  Moreover, at the beginning of the execution of \CinsertMP($z$) (except the first call),
  integer~$\pointer_z$ points to the position of the first agent in the list~$L_z$ that is strictly less preferred than all agents that were previously inserted into~$A$. 
  Now, whenever called, \CinsertMP($z$) uses $\pointer[z]$ to continue iterating through the preference list~$L_z$.
  Only unmatched agents~$L_z[\pointer[z]]$ (see Line~\eqref{insert:unmatched}) that are most-preferred by~$z$ are inserted into the matrix~$A$ by setting $A[z,y']=1$ until
  either the end of the the preference list~$L_z$ is reached~(i.e., $\pointer[z]=n+1$),
  or at least one agent was inserted (i.e., $B[z]>0$) such that $\pointer_z$
  points to the position of an agent that is strictly less preferred to the just inserted agent (stored as~$y'$).

  Now, to show the correctness of \cref{alg:BText}, we need to show 
  \begin{compactenum}[(i)]
    \item\label{correctness:1} whenever \cref{alg:BText} adds a pair~$\{x,y\}$ to~$M$, then $x$~and~$y$ indeed are each other's most acceptable agents (under the profile where all pairs previously added to $M$ are deleted),
    and 
    \item\label{correctness:2} there is always a possible pair in the queue~$Q$ where both agents are still unmatched. 
  \end{compactenum}  

  To verify~\eqref{correctness:1}, clearly, the first pair that is added to $M$ satisfies~\eqref{correctness:1}.
  Assume that $\{x_i,y_i\}$ is the $i^{\text{th}}$ pair added to $M$ such that all $i-1$ previous pairs~$q_1,\ldots, q_{i-1}$ added to $M$ satisfy \eqref{correctness:1}.
  Suppose, for the sake of contradiction, that $y_i$ is \emph{not} one of the most acceptable agents of $x_i$, after all agents in the pairs~$q_1,\ldots,q_{i-1}$ are deleted from the profile.
  Let $y$ be a not-yet-matched agent that is strictly preferred to~$y_i$ by $x_i$.
  Now, consider the call~\CinsertMP{} when $\{x_i,y_i\}$ is pushed to $Q$.  
  By \eqref{inv:find-next} and since $\pointer_z$ is none-decreasing,
  it follows that $\{x_i,y\}$ is added to~$Q$ in a previous round.
  This means that $\{x_i,y\}$ must be considered before $\{x_i,y_i\}$ in the queue~$Q$ and will thus be added to~$M$, a contradiction.
  
  As for \eqref{correctness:2}, observe that whenever the algorithm pops a pair from~$Q$,
  the queue~$Q$ contains every pair of agents that currently consider each other most acceptable (see \eqref{inv:find-next} and \eqref{inv:push-next}).
  The queue~$Q$ may contain further pairs of agents where one or both agents are already deleted; these pairs, however, will be detected (Line~\eqref{detect}) and ignored by the algorithm.
  Thus, the existence follows analogously to the one for \cref{alg:BT}.
  Finally, we analyze the running time.
  First of all, every agent~$x$ is deleted (by calling \CsmartD($x$)) at most once.
  Each call of \CsmartD($x$) takes $O(n)$~steps 
  (some constant number of operations, adjusting one row and one column in~$A$, each of size~$n$), 
  plus some calls of \CinsertMP{}.
  Second, \CinsertMP{} takes amortized~$O(n)$ running time for each agent~$z$ since it iterates through the preference list of~$z$ only once.
  Third, after the initialization, every entry~$A[x,y]$ is set from zero to one at most once~(see Line~\eqref{insert:set-one}), which happens only when $A[x,y]$ was zero.
  It is reset to zero at most once, which happens only when \CsmartD($x$) or \CsmartD($y$) is called and
  when $A[x,y]$ was one~(see Lines~\eqref{delete:set-zero1} and \eqref{delete:set-zero2}).
  Since $\{x,y\}$ is only pushed to $Q$ in \CinsertMP($x$) or \CinsertMP($y$)
  when  $A[x,y]=A[y,x]=1$ and $D[x]=D[y]=0$ (see \eqref{inv:z-not-deleted} in the above reasoning and Lines~\eqref{insert:unmatched}, \eqref{insert:set-one}, and \eqref{insert:push-if}),
  it follows that~$\{x,y\}$ can be pushed to the queues~$Q$ at most once.
  Since there are $O(n^2)$ pairs that can be added to $Q$ and since \CsmartD{} is only called whenever a pair is added to $M$, \CcomputeSM{} indeed takes $O(n^2)$~time.
\end{proof}

Now, we move on to (tie-sensitive) single-crossingness.
\begin{proposition}\label[proposition]{complete-sc:p}  
  If a preference profile with $2n$~agents is complete, with ties, narcissistic and single-crossing,
  then it always admits a stable matching,
  which can be found 
  in $O(n^2)$ time.
\end{proposition}


\begin{proof}
  By \cref{thm:complete:nsc->sp} and \cref{prop:relation-sc-tsc}\eqref{statement:noties:tiesc=sc}, the stated profiles are single-peaked. 
  By \cref{thm:complete-ties-nsp-qudratic}, we obtain the desired statement. 
\end{proof}

\section{Incomplete preferences}\label{sec:incomplete}
In this section, we consider the case when the input may contain incomplete preferences, meaning that the underlying acceptability graph may not be a complete graph. 
One reason for the occurrence of incomplete preferences could be that two agents may consider each other unacceptable and do not want to be matched together, or they are not ``allowed'' to be matched to each other.
If in this case no two agents are considered of equal value by any agent 
(i.e., the preferences do not have ties),
then \SR still remains polynomial-time solvable~\cite{GusfieldIrving1989}.
However, once ties are involved, \SR becomes NP-complete~\cite{Ronn1990} even for complete preferences.

First of all, we observe that once ties are allowed, 
neither single-peakedness nor single-crossingness, combined with narcissism, can guarantee that there are always two agents that are each other's most acceptable agent~(see \cref{prop:incomplete}\eqref{statement:inc:nspsc-no-existence}).
However, having two such agents is crucial for the existence of a stable matching so that it can be found by the algorithm by \citet{BartholdiTrick1986}~(see \cref{alg:BT}). 
Moreover, for incomplete preferences, even without ties, narcissistic and single-crossing preferences do not imply single-peakedness anymore.

\begin{proposition}\label[proposition]{prop:incomplete}
  For incomplete preferences without ties, the following holds.
  \begin{enumerate}[(1)] 
    \item\label{statement:inc:nsc-neq-nsp} Narcissistic and single-crossing preferences are \emph{not} necessarily single-peaked.
    \item\label{statement:inc:nspsc-no-unique} 
    Narcissistic and single-peaked (resp.\ single-crossing) preferences do not guarantee the uniqueness of stable matchings.
    \item\label{statement:inc:nspsc-no-existence} 
    Narcissistic and single-peaked (resp.\ single-crossing) preferences do not guarantee the existence of stable matchings.
  \end{enumerate}
\end{proposition}

\begin{proof}
  Statement~\eqref{statement:inc:nsc-neq-nsp}: Consider the following profile with six agents~$1,2,\ldots, 6$:
\begin{alignat*}{4}
   \text{Profile}~\Pot_1\colon 
  &\agent~1\colon &&1 \succ_1 3 \succ_1 4 \succ_1 2 \succ_1 5 \succ_1 6, \\
&    \text{agent }2\colon &&2 \succ_2 4 \succ_2 3 \succ_2 6 \succ_2 5 \succ_2 1, \\
&  \text{agent }3\colon &&3 \succ_3 2 \succ_3 1 \succ_3 6 \succ_3 5, \\
&  \text{agent }4\colon &&4 \succ_4 1 \succ_4 2, \\
  &  \text{agent }5\colon &&5  \succ_5 2 \succ_5 3 \succ_5 1,\\
&   \text{agent }6\colon &&6  \succ_6 2 \succ_6 3 \succ_6 1.
 \end{alignat*}
\noindent One can check that the profile is narcissistic, and it is single-crossing with respect to the order~$5 \rhd 6 \rhd 2 \rhd 3 \rhd 4 \rhd 1$.
But it is not single-peaked with respect to~$\rhd$ because agent~$1$'s preference order on~$\{1,3,4\}$ is not single-peaked with respect to~$\rhd$.
In fact, the profile is not single-peaked at all since the preference orders of agents~$1$,~$2$, and~$3$ restricted to the alternatives~$\{1,5,6\}$ form a forbidden subprofile (the so-called \emph{worst-configuration}) for the single-peaked property~\cite{BaHa2011}:
In a single-peaked order, regarding the preference orders of agents~$1$ and $2$, agent~$5$ \emph{must} be ordered between~$1$ and $6$, which is not possible because of the preference order of agent~$3$.  
The profile does not admit a perfect stable matching, i.e., a stable matching of size \emph{three}.
But it admits two stable matchings of size two each: $\{\{1,3\},\{2,4\}\}$ and $\{\{1,4\}, \{2,3\}\}$.

Statement~\eqref{statement:inc:nspsc-no-unique}: Consider the following profile with four agents~$1,2,3,4$.
\begin{alignat*}{4}
  \text{Profile~}\Pot_2\colon\quad &  \agent~1\colon &&1 \succ 2 \succ 3 \succ 4,\quad &~~&
    \agent~2\colon &&2 \succ 4 \succ 1, \\
  &  \agent~3\colon &&3 \succ 1 \succ 4, &&
    \agent~4\colon &&4 \succ 3 \succ 2 \succ 1.
 \end{alignat*}
Once can check that the profile is narcissistic and single-peaked with respect to the order~$1\rhd 2 \rhd 3 \rhd 4$,
and single-crossing with respect to the order $1\rhd' 3 \rhd' 2 \rhd' 4$.
It admits two different stable matchings $\{\match{}{1}{2}, \match{}{3}{4}\}$ and $\{\match{}{1}{3}, \match{}{2}{4}\}$.

Statement~\eqref{statement:inc:nspsc-no-existence}: Consider the following profile with six agents~$1,2,\ldots, 6$.
 \begin{alignat*}{4} 
   \text{Profile~}\Pot_3\colon\quad&  \agent~1\colon&& 1 \succ {5} \succ 2,\quad&  ~~& \agent~2\colon&& 2 \succ {1}   \succ 3,\\
 &  \agent~{3} \colon&& {3}\succ 2  \succ {4}, &\quad & \agent~{4} \colon&& 4 \succ {3} \succ 5,\\
  &  \agent~{5} \colon&& {5}\succ {4} \succ  1 \succ {6}, \quad && \agent~{6} \colon&& {6}\succ 5.
 \end{alignat*}
It is narcissistic, and single-peaked and single-crossing with respect to the order
$3 \rhd 2 \rhd 1 \rhd 4 \rhd 5 \rhd 6$. 
One can check that the profile is single-peaked with respect to the order~$\rhd$. 
For single-crossingness, observe that each pair of agents is ranked by at most two different agents so that the profile
is (tie-sensitive) single-crossing with respect to any ordering.

However, no matching~$M$ is stable for this profile.
To see this, notice that the preferences of agents from $A^*\coloneqq \{1,2,3,4,5\}$ form a certain cyclic structure:
For each agent~$i\in A^*$ it holds that agent~$i$ is the most preferred agent of agent~$(i\bmod 5)+1$.
Now, consider an arbitrary matching~$M$.
Since $|A^*|$ is odd, there is at least one agent~$i\in A^*$ with $M(i)\notin A^*$.
It is straightforward to see that agent~$i$ and agent~$(i\bmod 5)+1$ will form a blocking pair for~$M$.
\end{proof}

For the case when ties in the preferences are allowed, \citet{Ronn1990} showed that \SR becomes NP-hard even if the preferences are complete.
The constructed instances in his hardness proof, however, are not always single-peaked or single-crossing.
It is even not clear whether the problem remains NP-hard for this restricted case.
If the preferences need not be complete, however, 
then we can show NP-hardness,
by a completely different reduction, obtaining our main intractability result.

Before we state the corresponding theorem, we prove the following two lemmas which are used heavily
in our preference profile construction
to force two agents of specific types to be matched together.
The first lemma summarizes an observation on a profile that is similar 
to profile~$\Pot_3$ presented in the proof of \cref{prop:incomplete}.


\begin{lemma} \label[lemma]{lem:selector-gadget}
  Let $A=\{a^1,a^2,a^3, a^4, a^{5}\}$ be a set of five distinct agents
  and let $X$ be a non-empty set of agents disjoint from~$A$. 
  The preference orders of the agents in $A$ satisfy the following, where $(X)$ means that the agents in set~$X$ are tied with each other:
  \begin{alignat*}{7}
   &\agent~a^1\colon&& a^1 \succ a^{5}\succ a^2,&~&  \agent~a^{2} \colon&& a^2 \succ a^{1}\succ a^3, \\
   &\agent~a^3\colon&& a^3 \succ a^{2}   \succ a^4,&& \agent~a^{4} \colon&& a^{4}\succ a^{3}\succ  a^{5}, &~ &
   &\agent~a^{5} \colon&& a^{5}\succ (X) \succ a^{4} \succ a^1.
  \end{alignat*}
  Then, the following holds.
  \begin{enumerate}[(1)]
    \item\label{lem:e-selector}  Every stable matching~$M$ for~$A$ satisfies $M(a^{5})\in X$, $M(a^1)=a^2$, and $M(a^3)=a^4$.
    \item\label{lem:A-nsp} The preferences from~$A$ are narcissistic and single-peaked with respect to the linear order~$a^{3}\rhd a^{2} \rhd a^1 \rhd a^4 \rhd a^{5} \rhd [X]$,
    where $[X]$ denotes some fixed linear order of the agents in $X$.
    \item\label{lem:A-tiesc} The preferences from~$A$ are tie-sensitive single-crossing with respect to any ordering of the agents in $A$.
  \end{enumerate}
\end{lemma}
\begin{proof}
  Towards a contradiction to Statement~\eqref{lem:e-selector}, suppose that there is a stable matching~$M$ with $M(a^5)\notin X$. 
  Then, by the preferences of $a^5$, there remain three possibilities~(i)--(iii) for the partner of $a^5$.
  We claim to obtain a blocking pair for $M$ for each of these possibilities.
  
  \begin{compactenum}
    \item[\textbf{Case (i):}] $M(a^{5})\notin \{a^1, a^{4}\}$. This implies that agent~$a^{5}$ does not have a partner. Then, by construction, $\{a^{5},a^1\}$ will be blocking~$M$.
 
    \item[\textbf{Case (ii):}] $M(a^{5}) = a^{4}$. Since $a^4$ prefers $a^3$ to $a^5$, 
    it follows that $a^3$ must obtain a partner that it prefers to $a^4$ as, otherwise, $\{a^3,a^4\}$ will form a blocking pair for $M$.
    Since $a^2$ is the only (acceptable) agent that $a^3$ prefers to $a^4$, it follows that $\{a^{3}, a^{2}\}\in M$.
    Consequently, $a^1$ will not be assigned a partner by $M$. 
    However, $\{a^1, a^2\}$ will form a blocking pair for~$M$. 

  \item[\textbf{Case (iii):}] $M(a^{5}) = a^{1}$.
  Analogously to Case~(ii), if $M(a^{5})=a^1$, then we deduce that $M(a^{4})=a^3$.
  This implies that $a^2$ remains unmatched by $M$.
  However, $\{a^2,a^3\}$ will form a blocking pair for $M$.
 \end{compactenum}
 We have just shown that $M(a^5)\in X$. By the preferences of $a^2$ it follows that $M(a^2)=a^1$ as otherwise $a^1$ is unmatched and will form a blocking pair with~$a^2$. Similarly, we can infer that $M(a^4)=a^3$.

 Statement~\eqref{lem:A-nsp} is straight-forward.
 As for Statement~\eqref{lem:A-tiesc}, which concerns \tiescness{},
 we observe that it holds with respect to every order of the agents:
 Each pair of agents is ranked by at most two different agents and ties only occur in the preferences of~$a^5$.
\end{proof}

Whereas \cref{lem:selector-gadget} enforces that a specific agent~(i.e., $a^5$) must be matched with some agent
from a specific group of agents~(i.e., $X$),
the forthcoming lemma enforces some specific combination of matchings inside this specific group~$X$.
This seems crucial to ensure the narcissistic, single-peaked, and \tiesc property.

\begin{lemma}\label[lemma]{lem:vertex-gadget}
  Let $X\coloneqq \{x^1, x^2, \ldots, x^{10}\}$, $R\coloneqq \{\wred{r^2},\wred{r^4},\wred{r^6}, \wred{r^8}\}$ be two disjoint sets of agents, 
  and let $A$ and $B$ be two disjoint sets of agents such that $X, R, A$, and $B$ are pairwise disjoint.
  Furthermore, assume that the preferences of the agents from~$X$ satisfy the following, where the symbols~$[A]$ and $[B]$ denote some fixed linear orders of the agents in $A$ and $B$, respectively:
  \begin{alignat*}{4}
    \agent~x^1\colon&~ x^1 \succ x^{10} \succ [A] \succ x^2,&\quad&  \agent~x^{2} \colon&~& x^2 \succ x^{1}\succ \wred{r^2}\succ x^3. \\
   \agent~x^3\colon&~ x^3 \succ x^{2}   \succ x^4,&\quad& \agent~x^{4} \colon&~& x^{4}\succ x^{3}\succ \wred{r^4} \succ  x^{5}.\\  
   \agent~x^{5} \colon&~ x^{5}\succ x^4  \succ x^{6}. &&\agent~x^{6} \colon&~& x^{6}\succ x^5 \succ \wred{r^6} \succ  x^{7}.\\
   \agent~x^{7} \colon&~ x^{7}\succ x^6 \succ x^{8}.   &&\agent~x^{8} \colon&~& x^{8}\succ x^7 \succ \wred{r^8} \succ x^{9}. \\
   \agent~x^{9} \colon&~ x^{9}\succ x^8 \succ x^{10}.   &&\agent~x^{10} \colon&~& x^{10}\succ x^9 \succ [B] \succ x^{1}. 
  \end{alignat*}
  The following holds for the preferences of $X$.
  \begin{enumerate}[(1)]
    \item\label{lem:e-vertex}   
    If $M$ is a stable matching with $M(x^{10})\in B$, then for each~$i\in \{4,3,2,1\}$ we have that $\{x^{2i+1},x^{2i}\}\in M$.
    
    \item\label{lem:nspsc} The above preferences are narcissistic, single-peaked, and \tiesc with respect to the following order,
  where $\overleftarrow{[A]}$ and $\overleftarrow{[B]}$ denote the reverse of the fixed orders~$[A]$ and $[B]$, respectively:
  \[\overleftarrow{[A]}\succ \overleftarrow{[B]} \succ x^{10} \rhd x^9 \rhd x^1 \rhd x^2 \rhd \wred{r^2} \rhd x^3 \rhd x^4 \rhd \wred{r^4} \rhd x^5 \rhd x^6  \rhd \wred{r^6} \rhd x^7 \rhd x^8 \rhd \wred{r^8}.\]
  \end{enumerate}
\end{lemma}

\begin{proof}
  By our assumption that $M(x^{10})\in B$ and by the preference order of $x^{10}$, it follows that $M(x^9)=x^8$ as otherwise $\{x^9,x^{10}\}$ will form a blocking pair for~$M$.
  By an analogous reasoning for $x^7$, $x^5$, and $x^3$, we deduce that $M(x^7)=x^6$, $M(x^5)=x^4$, and $M(x^3)=x^2$, showing Statement~\eqref{lem:e-vertex}.

  As for Statement~\eqref{lem:nspsc}, clearly, the preferences of the agents in $X$, are narcissistic. 
  One can check that the preferences are indeed single-peaked with respect to the order~$\rhd$.
  Since no two agents are ranked by more than two different agents, it is clear that the preferences are \tiesc,
  implying single-crossingness due to \cref{prop:relation-sc-tsc}\eqref{statement:noties:tiesc=sc}. 
\end{proof}

\cref{fig:gadget} illustrates the crucial part of the acceptability graphs of the profiles discussed in \cref{lem:selector-gadget} and  \cref{lem:vertex-gadget}.
Observe that, these graphs display a certain cyclic structure if we make the 
edges directed: 
If we construct an arc~$(u,w)$ for each agent~$u$ and its most acceptable agent~$w$,
then we obtain a directed cycle.
Moreover, we have already seen from profile~$P_3$ in the proof of \cref{prop:incomplete} that odd cycles imply non-existence of stable matchings.
Using this observation, we can show our main result now.

\begin{figure}
  \centering
     \def \nodesize {20pt}

      \tikzstyle{pnode} = [fill=black!10]
      \tikzstyle{cnode} = []
      \tikzstyle{match} = [line width = 2.5pt, red]
      \tikzstyle{secondmatch} = [line width = 1.2pt, double, double distance=2pt]
      \begin{tikzpicture}[every node/.style={draw=black,thick,circle,inner sep=1pt, font=\footnotesize}]
      \def \n {5}    
      \def \radius {.9cm}
      \def \margin {12}
      
      \foreach \s in {1,2,3,4,5} {
        \node[draw, minimum size=\nodesize] at ({360/\n * \s}:\radius) (a\s) {$a^\s$};
      }
      
      \node[draw,pnode,minimum size=\nodesize, right = .5cm of a5] (X) {$X$};

      \foreach \i /\j in {1/2,2/3,3/4,4/5,5/1} {
        \draw (a\i) edge[normalline] (a\j);
      }
      \foreach \i /\j in {a1/a2,a3/a4,a5/X} {
        \draw (\i) edge[match] (\j);
      }

    \end{tikzpicture}
    ~\qquad~
    \begin{tikzpicture}[every node/.style={draw=black,thick,circle,inner sep=1pt, font=\footnotesize}]
      \def \n {5}
      \def \yzero {2.4}
      \def \yone {1.2}
      \def \ythree {0}
      \def \yfour {-1.2}
      \def \xone {0}
      \def \xtwo {6.3}

      \foreach \s / \i / \j  in {1/10/\ythree,1/1/\yone,2/2/\yone,3/3/\yone,4/4/\yone,5/5/\yone,%
        5/6/\ythree, 4/7/\ythree,3/8/\ythree,2/9/\ythree} {
      \node[draw, pnode][minimum size=\nodesize] at ({\xtwo/\n * (\s)}, \j) (x\i) {$x^{\i}$};
    }

    \foreach \s / \x / \y / \name / \id in {0/\xtwo/\ythree/{B}/AS, 0/\xtwo/\yone/{A}/A, 2/\xtwo/\yzero/{r^2}/r2, 4/\xtwo/\yzero/{r^4}/r4, 5/\xtwo/\yfour/{r^6}/r6, 3/\xtwo/\yfour/{r^8}/r8} {
    \node[draw, cnode][minimum size=\nodesize] at ({\x/\n * \s}, \y) (\id) {$\name$};
  }

    \foreach \i / \j in {x1/x2,x2/x3,x3/x4,x4/x5,x5/x6,x6/x7,x7/x8,x8/x9,x9/x10,x10/x1,AS/x10,A/x1,x2/r2,x4/r4, r6/x6, r8/x8} {
      \draw (\i) edge[normalline] (\j);
    }
    \foreach \i / \j in {x2/x3,x4/x5,x6/x7,x8/x9,AS/x10,A/x1} {
      \draw (\i) edge[match] (\j);
    } 
    \foreach \i / \j in {x2/x1,x4/x3,x6/x5,x8/x7,x10/x9} {
      \draw (\i) edge[secondmatch] (\j);
    }
    \end{tikzpicture}

    \caption{Illustration of the acceptability graphs for \cref{lem:selector-gadget,lem:vertex-gadget}. 
      Left: The node labeled with $X$ represents the agents in $X$. Thick red lines correspond to a possible stable matching. 
      Right: The nodes labeled with $A$ and $B$ represent the agents in $A$ and $B$, respectively. 
      There are two possible stable matchings, one represented by thick red lines and the other by double lines.}\label{fig:gadget}
\end{figure}
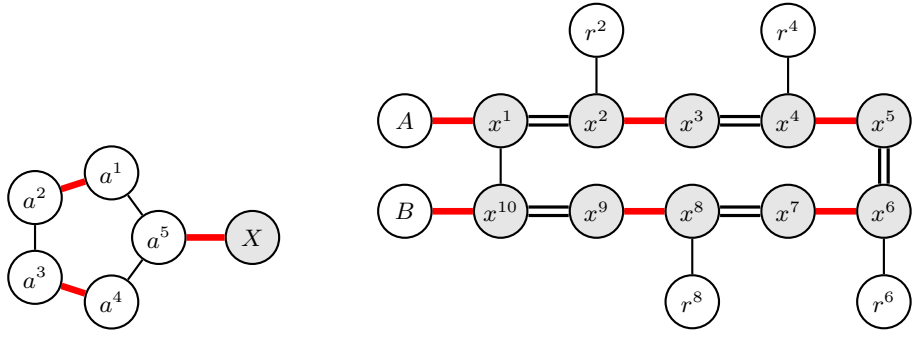

\begin{theorem}\label[theorem]{incomplete-ties-sp-sc:np-c}
  \SR for incomplete preferences with ties remains NP-complete, even if 
  the preferences are narcissistic, single-peaked, and (tie-sensitive) single-crossing.
\end{theorem}

\begin{proof}
  First, the problem is in NP since one can non-de\-ter\-min\-isti\-cally guess a matching and check stability in polynomial time.
  To show NP-hardness, we reduce from the NP-complete \textsc{Independent Set} problem~\cite{GJ79}:

  \decprob{Independent Set}{
    an undirected graph~$G=(U,E)$ and a non-negative integer~$k$.}
  {Is there a size-$k$ \emph{independent set}, i.e., 
  a subset~$U'\subseteq U$ of $k$~pairwise non-adjacent vertices?}
  We assume that each vertex has degree at most three since \textsc{Independent Set} remains NP-hard for this case~\cite{GJS76}.
  Let $(G=(U,E),k)$ be an \textsc{Independent Set} instance with $U$ being the vertex set and $E$ being the edge set.
  Let $U\coloneqq \{u_1, u_2,\dots, u_n\}$.
  We will construct a \SR instance~$\Pot$ with agent set~$V$ which is narcissistic, single-peaked, and \tiesc,
  and show that~$G$ admits a size-$k$ independent set if and only if~$\Pot$ admits a stable matching.
  
  \paragraph{Main idea and the constructed agents.}
 
  For each vertex~$u_i \in U$, we introduce ten \emph{vertex agents}~$u^j_i$, $1\le j \le 10$, denoted as $U_i=\{u_i^j\mid 1\le j \le 10\}$.
  The idea is to obtain an acceptability graph that has a cycle of length ten (see our discussion just prior to the theorem) for each vertex of the input graph, including eight for the neighbors of this vertex.
  This is used to ensure narcissistic, single-peaked, and single-crossing property simultaneously. 
  Additionally, we introduce two groups of \emph{selector agents}, each with $k$ sets:
  \begin{compactenum}
    \item[\textbf{Group (1):}] $A_j=\{a^i_j\mid 1\le i \le 5\}$, $1\le j \le k$, and
    \item[\textbf{Group (2):}] $B_j=\{b^i_j\mid 1\le i \le 5\}$, $1\le j \le k$.
  \end{compactenum}
  We will construct preferences for these selector agents to enforce that 
  each two selectors $a_j^5$ and $b^5_j$, $1\le j \le k$, are matched to two vertex agents.
  Together with the preferences of the vertex agents, we make sure that the vertex agents that are matched to $a_j^5$ and $b^5_j$, respectively, correspond to the same vertex.

  The agent set~$V$ is defined as $V=\left(\bigcup_{1\le i \le n}U_i\right) \cup \left(\bigcup_{1\le j\le k}(A_j \cup B_j)\right)$.
  In total, we have constructed $10n+10k$~agents.

  To encode an \textsc{Independent Set} instance, we aim to construct preferences for the vertex agents such that \emph{no} two vertex agents that are matched to some selector agents correspond to two adjacent vertices.
  This property will be formally captured later by \cref{claim:stable-matching-properties}.

\paragraph{Selector agents.} The preferences of the selector agents are of the form as described in \cref{lem:selector-gadget}.
To this end, let $U^1=\{u^1_i\mid 1\le i \le n\}$ and $U^{10}=\{u^{10}_i\mid 1\le i \le n\}$.
We use $(U^1)$ and $(U^{10})$ to express that the agents in the respective subsets are tied.
\allowdisplaybreaks
\begin{alignat*}{4}
   \intertext{ $\forall j \in \{1,2,\dots, k\}\colon$ } & \text{selector agent}~a_{j}^1\colon &&  a_{j}^1 \succ a_{j}^5 \succ a_{j}^2,  &\quad &  \text{selector agent}~b_{j}^1\colon &&  b_{j}^1 \succ b_{j}^5 \succ b_{j}^2,\\
    & \text{selector agent}~a_{j}^2\colon &&  a_{j}^2 \succ a_{j}^1 \succ a_{j}^3,    &\quad & \text{selector agent}~b_{j}^2\colon &&  b_{j}^2 \succ b_{j}^1 \succ b_{j}^3,  \\
    & \text{selector agent}~a_{j}^3\colon && a_{j}^3 \succ a_{j}^2 \succ a_j^4,  &\quad   & \text{selector agent}~b_{j}^3\colon && b_{j}^3 \succ b_{j}^2 \succ b_j^4, \\
    & \text{selector agent}~a_{j}^4\colon &&  a_{j}^4 \succ a_{j}^3 \succ a_j^5,     & \quad & \text{selector agent}~b_{j}^4\colon &&  b_{j}^4 \succ b_{j}^3 \succ b_j^5, \\
    & \text{selector agent}~a_{j}^5\colon &&  a_{j}^5 \succ (U^1) \succ a_{j}^4 \succ a_j^1.& \quad & \text{selector agent}~b_{j}^5\colon &&  b_{j}^5 \succ (U^{10}) \succ b_{j}^4 \succ b_j^1. 
\end{alignat*}

\renewcommand{\Agent}{\text{\normalfont vertex agent}}
  \paragraph{Vertex agents.} The preferences of the vertex agents for each vertex~$u_i$ are of the form described in \cref{lem:vertex-gadget}.
\begin{alignat*}{4}
  \intertext{$\forall i \in \{1,2,\dots, n\}\colon$}
&   \Agent~u_i^1\colon&~& u_i^1 \succ u_i^{10} \succ [A^5] \succ u_i^2,&\quad&  \Agent~u_i^{2} \colon&~& u_i^2 \succ u_i^{1}\succ \wred{r_i^2}\succ u_i^3, \\
  & \Agent~u_i^3\colon&~& u_i^3 \succ u_i^{2}   \succ u_i^4,&& \Agent~u_i^{4} \colon&~& u_i^{4}\succ u_i^{3}\succ \wred{r_i^4} \succ  u_i^{5},\\  
  & \Agent~u_i^{5} \colon&~& u_i^{5}\succ u_i^4  \succ u_i^{6}, &&\Agent~u_i^{6} \colon&~& u_i^{6}\succ u_i^5 \succ \wred{r_i^6} \succ  u_i^{7},\\
   &\Agent~u_i^{7} \colon&~& u_i^{7}\succ u_i^6 \succ u_i^{8},  &&\Agent~u_i^{8} \colon&~& u_i^{8}\succ u_i^7 \succ \wred{r_i^8} \succ u_i^{9},\\
  & \Agent~u_i^{9} \colon&~& u_i^{9}\succ u_i^8 \succ u_i^{10},  &&\Agent~u_i^{10} \colon&~& u_i^{10}\succ u_i^9 \succ [B^5] \succ u_i^{1}. 
\end{alignat*}
\noindent Herein, $A^5$ represents the following set~$A^5\coloneqq \{a^5_j\mid 1\le j \le k\}$
and $[A^5]$ denotes the following fixed order of the selector agents in $A^5$: $[A^5]\coloneqq a^5_1 \succ a^5_2\succ \cdots \succ a^5_k$.
The symbols~$B^5$ and $[B^5]$ are defined in an analogous way: 
$B^5\coloneqq \{b^5_j\mid 1\le j \le k\}$ and $[B^5]$ denotes the following fixed order of the selector agents in $B^5$: $[B^5]\coloneqq b^5_1 \succ b^5_2\succ \cdots \succ b^5_k$.

It remains to specify the agents~$r_i^2$, $r_i^4$, $r_i^6$, and $r_i^8$.
Recall that the maximum vertex degree of our input graph~$G$ is three. 
By Vizing's Theorem~\cite{vizing1964estimate}, graph~$G$ is 4-edge-colorable, that is, the edge set of $G$ can be partitioned into four subsets~$E_1$, $E_2$, $E_3$, and $E_4$ which are each a matching.
Moreover, this partition can be computed in polynomial time~\cite{misra1992constructive}.
Thus, we first compute the sets~$E_c$, $1\le c \le 4$.
Next, for each $1\le c \le 4$, if $E_c$ contains an edge~$e=\{u_i,u_j\}$,
then we define $r^{2c}_i:=u^{2c}_j$ (resp.\ $r^{2c}_j\coloneqq u^{2c}_i$).
If some vertex~$u_i$ is not incident to an edge of color~$c$, then $r^{2c}_i$ is omitted, that is, the preference order of agent~$u^{2c}_i$ is just~$u_i^{2c} \succ u_i^{2c -1}  \succ  u_i^{2 c +1}$.

For an illustration, assume that~$e_1\in E_1$ and $e_4\in E_2$ with $e_1=\{u_1, u_2\}$ and $e_4=\{u_1,u_3\}$.
Then, $r_1^2,r_2^2,r_1^4,r_3^4$ are $r_1^2=u_2^2$, $r_2^2=u_1^2$, $r_1^4=u^4_3$, $r_3^4=u_1^4$.

  This completes the construction, which can clearly be performed in polynomial~time.
  Next, we show that our constructed profile is narcissistic, single-peaked, single-crossing, and \tiesc.
 \newcommand{\spo}{\rhd}

\paragraph{Single-peakedness.} 
  The con\-struc\-ted profile is single-peaked with respect to the following linear order~$\spo$:
  \begin{align*}
    [A_k] \spo [A_{k-1}] \spo \cdots \spo [A_1] \spo [B_k] \spo [B_{k-1}] \spo \cdots \\ \cdots \spo [B_1]\spo [U^{10,9}] \spo [U^{1,2}] \spo [U^{3,4}] \spo [U^{5,6}] \spo [U^{7,8}] \text{.}
  \end{align*}
  We specify the notations used in the above order.
  \begin{compactitem}[--]
    \item For each $j \in \{1,2,\dots, k\}$ and for each $(Q,q)\in \{(A,a),(B,b)\}$, the let~$[Q_j]$
          denote the order~$q_j^3\spo q_j^2 \spo q_j^1 \spo q_j^4 \spo q_j^5$ (cf.\ \cref{lem:selector-gadget}).
  \item The symbol~$[U^{10,9}]$ denotes the order~$u^{10}_1\spo u^{10}_2 \spo \cdots \spo u^{10}_n \spo u^{9}_1\spo u^{9}_2 \spo \cdots \spo u^{9}_n$.
  \item For each $c\in \{1,2,3,4\}$ the symbol~$U^{2 c-1, 2 c}$ 
  denotes the subset~$\{u_i^{2c-1}, u_i^{2c}\mid 1\le i \le n\}$.
  Intuitively, the symbol~$[U^{2 c-1, 2 c}]$ denotes an order of the agents in $U^{2 c-1, 2 c}$,
  which makes sure that the preferences of the vertex agents that are ``incident'' to the edges in $E_c$ are single-peaked with respect to this order.
  To define~$[U^{2 c-1,2 c}]$, we need the following additional notion:
  For each edge~$e\in E_c$ let $u_i$ and $u_j$ be the respective endpoints with $i<j$.
 Then, let $[e]$ denote the order~$u^{2 c -1}_i \spo u^{2 c}_i \spo u_j^{2 c-1} \spo u_j^{2 c}$.

 The order $[U^{2 c -1, 2 c}]$ is defined as follows:
 \begin{align*}
  [U^{2 c-1,2 c}] \coloneqq  [e_{c,1}] \spo [e_{c,2}] \spo \cdots \spo [e_{c,\ell}] \spo [R],
 \end{align*}
\noindent where $e_{c,1}, e_{c,2}, \dots, e_{c,\ell}$ is an arbitrary but fixed order of the edges in $E_c$, and $[R]$ denotes a fixed order of the vertex agents~$u_i$ of the form~$u_i^{2 c-1}\succ u_i^{2 c}$ that are not ``incident'' to any edge in $E_c$.
\end{compactitem}
For an illustration, assume that~$E_1=\{e_1,e_2\}$ with $e_1=\{u_1, u_2\}$ and $e_2=\{u_5,u_6\}$.
Then, $[U^{1,2}]$ may be defined as $[U^{1,2}]\coloneqq u_1^{1}\spo u_1^2 \spo u_2^1 \spo u^2_2 \spo u^{1}_5 \spo u^2_5 \spo u^1_6 \spo u^2_6 \spo u^1_3 \spo u^2_3 \spo u_4^1 \spo u^2_4$.

\newcommand{\SCo}{\rhd}

\paragraph{Tie-single-crossingness and single-crossingness.}
\newcommand{\A}[1]{\mathsf{Ac}{#1}}
To show that the constructed preference profile is also \tiesc we consider each pair~$p$ of agents and let $\A(p)$
denote the agents 
which consider both agents in~$p$ as acceptable partners.
We show that either all agents in $\A(p)$ have the same order on $p$
or if two agents exist that order $p$ differently, then $|\A(p)|=2$. 
It is straight-forward to verify that if we can show the above statement, then
the profile is \tiesc, and it is single-crossing by breaking ties in an arbitrary but fixed way.

To this end, let $A^5=\{a_j^5\mid 1\le j \le k\}$, 
$B^5=\{b_j^5\mid 1\le j \le k\}$, 
$U^1=\{u_i^1\mid 1\le i \le n\}$, and $U^{10}=\{u_i^{10} \mid 1\le i \le n\}$.

 \begin{compactenum}
   \item[\textbf{Case 1:~}$\boldsymbol{p=\{a_j^{z},a_{j}^{z'}\}\subseteq A_j}$ (resp.\ $\boldsymbol{p=\{b_j^{z}, b_{j}^{z'}\}\subseteq B_j}$)] for some value~$j\in \{1,2,\dots, k\}$ and some values~$z,z'\in \{1,2,3,4,5\}$.
   This implies that $\A(p)\subseteq A_j$ (resp.\ $A(p)\subseteq B_j$). 
   Moreover, no agent in~$\A(p)$ 
   considers the agents in~$p$ to be tied with each other. 
   If there are two agents in~$\A(p)$ that order this pair differently, then these agents must be those from~$p$.
   It is straightforward to verify that in this case, it holds that $p=\A(p)$, implying that $|\A(p)|=2$. 
   For instance, if $p=\{a^1_i,a_i^5\}$, then $\A(p)=p$.
   Otherwise, $|A(p)|=1$ (For instance, for $p=\{a^2_i, a^5_i\}$, one can check that $\A(p)=\{a^1_i\}$).
   \item[\textbf{Case 2:~}$\boldsymbol{p=\{a_j^{z},a_{j'}^{z'}\}}$ (resp.\ $\boldsymbol{p=\{b_j^{z},b_{j'}^{z'}\}}$)]
   for two distinct values~$j,j'\in \{1,2,\cdots, k\}$ with $j\neq j'$
   and for some values~$z,z'\in \{1,2,3,4,5\}$.
   If $z = z' =  5$, then $\A(p)=U^1$ (resp.\ $\A(p)=U^{10}$); otherwise $\A(p)=\emptyset$.
   For the first case, all agents in $\A(p)$ have the same order on $p$.
   
   \item[\textbf{Case 3:~}$\boldsymbol{p=\{u_i^{j}, u_i^{j'}\} \subseteq U_i}$] for some value~$i\in \{1,2,\dots, n\}$ and some distinct values~$j,j'\in \{1,2,\dots, 10\}$.
   This implies that $\A(p)\subseteq U_i$,
   and no agent in $\A(p)$ considers the agents in $p$ to be tied with each other.
   Analogously, we deduce that if there are two agents in $\A(p)$ that order this pair differently, then they must be the agents in $p$ themselves.
   It is straightforward to verify that in this case we have $p=\A(p)$.
   Otherwise, $|\A(p)|=1$.


   \item[\textbf{Case 4:~}$\boldsymbol{p=\{u_{i}^{j}, u_{i'}^{j'}\}}$] for two distinct values~$i,i'\in \{1,2,\dots, n\}$
   with $i\neq i'$ and some two (possibly identical) values~$j, j'\in \{1,2,\ldots, 10\}$.
   If $j=j'\in \{1,10\}$, then either $\A(p)=A^5$ or $\A(p)=B^5$.
   In both cases, all agents in $\A(p)$ consider the agents in $p$ to be tied with each other.
   Otherwise, if the corresponding two vertices~$v_i$ and $v_{i'}$ are adjacent with $\{v_i,v_{i'}\}\in E_c$ for some color~$c\in \{1,2,3,4\}$
   such that $j=j'=2c$,
   then $\A(p)=p$.
    In the remaining cases for $j$ and $j'$, no agent considers~$u_i^{j}$ and $u^{j'}_{i'}$ both acceptable as partners and thus, we deduce that 
   $\A(p)=\emptyset$.

   \item[\textbf{Case 5:~}$\boldsymbol{p=\{u_i^z, a_{j}^{z'}\}}$ (resp.\ $\boldsymbol{p=\{u_i^z, b_{j}^{z'}\}}$)] for some values~$i,j,z,z'$ with $i \in \{1,2,\dots, n\}$, $j\in \{1,2,\dots, k\}$, 
   $z \in \{1,2,\ldots, 10\}$, and $z'\in \{1,2,\ldots, 5\}$. 
   If $z=1$ (resp.\ $z=10$) and $z'=5$, then by the constructed preferences we deduce that $\A(p)=p$.

   If $z=1$ (resp.\ $z=10$) and $z'\in \{1,4\}$, then by the constructed preferences we deduce that $\A(p)=\{a^5_j\}$ (resp.\ $\A(p)=\{b^5_j\}$).

   If $z\in \{2,10\}$ (resp.\ $z \in \{1,9\}$) and $z'=5$, then we deduce that $\A(p)=\{u^1_i\}$ (resp.\ $\A(p)=\{u^{10}_i\}$).

   Otherwise, by construction, we have that $\A(p)=\emptyset$.
 \end{compactenum}

  \paragraph{Correctness of the construction.}
  We show that~$G$ has a size-$k$ independent set if and only if $\Pot$ admits a stable matching.
  
  For the ``only if'' part, assume that $U'\subseteq U$ is a size-$k$ independent set 
  where $U'=\{u_{q_1},u_{q_2},\ldots, u_{q_k}\}$ with $q_{i-1}< q_{i}$ for all $i\in \{2,3,\dots, k\}$.
  One can verify that the following matching~$M$ is stable.
  \begin{compactitem}[--]
    \item For each $u_{q_j} \in U'$ set $M(a^5_j) = u^1_{q_j}$, $M(b^5_j) = u^{10}_{q_j}$, $M(a^{1}_j)=a^{2}_j$, $M(a^{3}_j)=a^4_j$,  $M(b^{1}_j)=b^{2}_j$, and $M(b^{3}_j)=b^4_j$.
    Note that this in particular defines matchings for all selector agents.
    \item For each $u_i\in U'$ and for each $c\in \{1,2,3,4\}$, set $M(u_i^{2c+1})=u_i^{2c}$.
    \item For each $u_i \notin U'$ and for each $c\in \{1,2,3,4,5\}$ set $M(u^{2c-1}_i) = u^{2c}_i$.
  \end{compactitem}

   Before we go on with the ``if'' part, we observe some properties that each stable matching of our constructed profile must satisfy.
  \begin{myclaim}\label[myclaim]{claim:stable-matching-properties}
    Every stable matching~$M$ for our constructed profile
    satisfies the following two properties.
    \begin{compactenum}[(1)]
      \item \label{prop:vertices-matched}
      For each~$j\in \{1,2,\ldots, k\}$ 
      the selector agent~$b^5_j$ is matched with some vertex agent of the form~$u^{10}_i$ for some $i\in \{1,2,\ldots, n\}$.
      
      \item \label{prop:edges-covered}
      \emph{No} two vertex agents~$u^{10}_i$ and $u^{10}_{i'}$ which are both matched to some selector agents correspond to two adjacent vertices.
    \end{compactenum}
  \end{myclaim}

  \begin{proof}[Proof of \cref{claim:stable-matching-properties}] 
    \renewcommand{\qed}{\hfill~(of
      \cref{claim:stable-matching-properties})~$\diamond$}
    Let $M$~be a stable matching for our profile.
    For the first statement, \cref{lem:selector-gadget}(\ref{lem:e-selector}) immediately implies that 
    for every selector agent~$b^5_j$, $1\le j \le k$, it holds that $M(a^5_j) \in \{u^{10}_i \mid 1\le i \le n\}$. 

    For the second statement, 
    suppose, towards a contradiction, that there are two vertex agents~$u^{10}_i, u^{10}_{i'}$ (for some distinct~$i,i'\in \{1,2,\ldots,n\}$) which are matched to some selector agents such that the corresponding vertices~$u_i$ and $u_{i'}$ are adjacent, 
    i.e., $\{u_i,u_{i'}\}\in E_c$ for some $c\in \{1,2,3,4\}$.
    This means that $\{M(u^{10}_i),M(u^{10}_{i'})\}\subseteq \{b^5_j\mid 1\le j \le k\}$.
    Consequently, by \cref{lem:vertex-gadget}\eqref{lem:e-vertex}, it follows that $M(u_i^{2c})=u_i^{2c+1}$ and $M(u_{i'}^{2c})=u_{i'}^{2c+1}$.
    However, by construction, both $u^{2c}_{i}$ and $u^{2c}_{i'}$ prefer each other to their respective partner, forming a blocking pair for~$M$---a contradiction.
    \qed
  \end{proof}

\noindent  To show the ``if'' part, let $M$ be a stable matching. 
  Then, \cref{claim:stable-matching-properties} immediately implies that the subset~$U'=\{u_i \mid M(u^{10}_i)=b^{5}_j \text{ for some } j \in \{1,2,\cdots, k\}\}$ is an independent set of size~$k$. 
 %
\end{proof}

\section{Conclusion}\label{sec:conclude}
We investigated \SR for preferences with popular structural properties, such as being narcissistic, single-peaked, and single-crossing.
We showed that for complete preferences with ties, assuming narcissism and single-peakedness (or narcissism and single-crossingness) guarantees the existence of stable matchings, which can be found in linear time; by comparison, \SR with complete preferences and ties is NP-complete~\cite{Ronn1990}.
/bin/bash: a: command not found
For incomplete preferences with ties, however, narcissistic and single-peakedness even combined with single-crossingness do not help to lower the computational complexity---\SR remains NP-complete.

We conclude with some challenges for future research. 
First, to better understand the NP-completeness result,
one may study the parameterized complexity with respect to the ``degree'' of incompleteness of the input preferences, such as the number of ties or the number of agents that are in the same equivalence class of the tie-relation. 
We refer to some recent papers on the parameterized complexity of preference-based stable matching
problems~\cite{MarxSchlotter2010,MarxSchlotter2011,MS20,CHSYicalp-par-stable2018,GuRoSaZe2017,MR18,AdGuRoSaZe2018,CheNieSkoECmstable2018,ChenSkowronSorge2019ec-robust,BreHeeKnoNie2019,BreCheKnoLuoNie2020aaai} for this line of research.

Second, we were not able to settle the computational complexity for complete preferences that are also single-peaked and single-crossing. 

Third, for incomplete preferences, we extended the concepts of single-peaked and single-crossing preferences. 
There are, however, further relevant extensions in the literature~\cite{Lack2014,ElkFalLacObr2015,FiHe2016}, which deserve attention in the study of stable matching problems.
Do our results transfer to preferences with these extensions?

Fourth, the algorithm of \citet{BartholdiTrick1986} strongly relies on the following condition:
\begin{quote}
  \emph{  At each step, there are always two agents that consider each other most acceptable.}
\end{quote}
To explore the boundary between tractable and intractable cases it is important to understand which structural properties imply this condition.
For instance, the so-called worst-restricted property (i.e., no three agents exist such that each of them is least preferred by 
some agent)~\cite{Sen1966,SePa1969}
is a generalization of the single-peaked property.
In fact, we can show that for preferences without ties, the narcissistic and worst-restricted properties
are enough to guarantee that there are always two agents which consider each other most acceptable~\cite{Finnendahl16}.
For profiles with ties, however, it seems unclear how to define
a worst-restricted property
so that the question of ``how far one can generalize structured preferences so that the above condition still holds'' remains open.

Finally, we found an issue with Bartholdi III and Trick's algorithm regarding their claimed~$O(n)$ running time~(see \cref{ex:BT-n^2}).
This brings up the question as to whether finding a stable matching can be done faster than in linear time, i.e., faster than $O(n^2)$~time for complete, narcissistic, and single-peaked preference profiles without ties.

\section*{Acknowledgment}
We thank anonymous reviewers of the \emph{5th International Conference on Algorithmic Decision Theory (ADT '17)} and of the journal \emph{Autonomous Agents and Multi-Agent Systems (JAAMAS)} for their constructive feedback.
In particular, we thank an anonymous reviewer of JAAMAS for pointing us to an issue
in the algorithm of Bartholdi III and Trick that also applied to our
algorithms.

Jiehua Chen was partially supported by the People Programme (Marie Curie Actions) of the European Union's Seventh Framework Programme (FP7/2007-2013) under REA grant agreement number {631163.11}, by the Israel Science Foundation (grant no. 551145/14), by European Research Council (ERC) under the European Union’s Horizon 2020 research and innovation programme under grant agreement numbers~677651, and by the WWTF research grant (VRG18-012). 

\end{document}